\documentclass{article}

\usepackage{arxiv}

\usepackage[utf8]{inputenc} 
\usepackage[T1]{fontenc}    
\usepackage{hyperref}       
\usepackage{url}            
\usepackage{booktabs}       
\usepackage{amsmath, amssymb}
\usepackage{amsfonts}       
\usepackage{nicefrac}       
\usepackage{microtype}      
\usepackage{lipsum}
\usepackage{amscd}
\usepackage{multicol}
\usepackage{graphicx}
\usepackage[T1]{fontenc}
\usepackage[utf8]{inputenc}
\usepackage{authblk}
\usepackage[ruled,vlined]{algorithm2e}
\usepackage{amsthm}

\newtheorem{theorem}{Theorem}[section]
\newtheorem{lemma}[theorem]{Lemma}
\newtheorem{proposition}[theorem]{Proposition}
\newtheorem{corollary}[theorem]{Corollary}
\newtheorem{definition}[theorem]{Definition}
\newtheorem{problem}[theorem]{Problem}

\title{Quadrilateral Mesh Generation II : Meromorphic Quartic Differentials and Abel-Jacobi Condition}
\author[1]{Na Lei}
\author[1]{Xiaopeng Zheng}
\author[1]{Zhongxuan Luo}
\author[2]{Feng Luo}
\author[3]{Xianfeng Gu}
\affil[1]{Dalian University of Technology}
\affil[2]{Rutgers University}
\affil[3]{Stony Brook University}

\begin{document}
\maketitle
\begin{abstract}%
This work discovers the equivalence relation between quadrilateral meshes and meromorphic quartic differentials. Each quad-mesh induces a conformal structure of the surface, and a meromorphic quartic differential, where the configuration of singular vertices correspond to the configurations of the poles and zeros (divisor) of the meroromorphic differential. Due to Riemann surface theory, the configuration of singularities of a quad-mesh satisfies the Abel-Jacobi condition. Inversely, if a divisor satisfies the Abel-Jacobi condition, then there exists a meromorphic quartic differential whose divisor equals to the given one. Furthermore, if the meromorphic quadric differential is with finite trajectories, then it also induces a a quad-mesh, the poles and zeros of the meromorphic differential correspond to the singular vertices of the quad-mesh.

Besides the theoretic proofs, the computational algorithm for verification of Abel-Jacobi condition is also explained in details. Furthermore, constructive algorithm of meromorphic quartic differential on genus zero surfaces is proposed, which is based on the global algebraic representation of meromorphic differentials.

Our experimental results demonstrate the efficiency and efficacy of the algorithm. This opens up a novel direction for quad-mesh generation using algebraic geometric approach.
\end{abstract}

\keywords{Quadrilateral Mesh \and Flat Riemannian Metric \and Geodesic \and Discrete Ricci flow \and Conformal Structure Deformation}

\section{Introduction}


\begin{figure}[h!]
\centering
\begin{tabular}{ccc}
\includegraphics[height=0.3\textwidth]{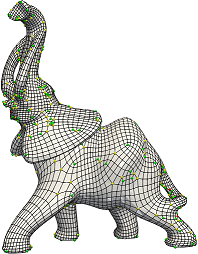}&
\includegraphics[height=0.3\textwidth]{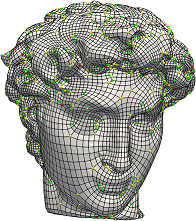}&
\includegraphics[height=0.3\textwidth]{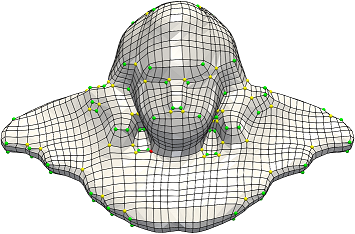}\\
\end{tabular}
\caption{Quadrilateral meshes with singularities on them, green, yellow and red dots represent singularities with valence $3$, $5$ and $6$ respectively.
\label{fig:quad_meshes}}
\end{figure}

\subsection{Motivation} Quadrilateral meshes play a fundamental role in computational mechanics, geometric modeling, computer aided design, animation, digital geometry processing and many fields. Despite of tens of years' research efforts, rigorous and automatic algorithms to produce high quality quad-meshes have not been achieved yet. The theoretic understanding of the singularity configurations (the locations and valences) still remains preliminary. This work focuses on giving sufficient and necessary conditions for singularity configurations based on Riemann surface theory.

More specifically, given a closed surface $\Sigma$ embedded in the Euclidean space $\mathbb{R}^3$, it has the induced Euclidean Riemannian metric $\mathbf{g}$. Suppose the surface is tessellated by a quadrilateral mesh $\mathcal{Q}$. If a vertex in $Q$ with topological valence $4$, then the vertex is \emph{regular}, otherwise \emph{singular}. The singularity configuration of $Q$ is represented as the \emph{divisor} of $Q$, defined as
\begin{equation}
    D_Q := \sum_{v\in Q} (k(v)-4) v,
    \label{eqn:divisor}
\end{equation}
where $v$ is a vertex of $Q$, $k(v)$ is the topological valence of $v$. The goal of this work is to find the necessary and sufficient conditions for quad-mesh divisors.

\subsection{Quad-mesh induced Structures}

A quad-mesh $\mathcal{Q}$ naturally induces several structures, each of them gives some information about the singularity configuration. The induced conformal structure gives the most fundamental and complete information.

\paragraph{Combinatorial structure} Suppose the number of vertices, edges, faces of $\mathcal{Q}$ are $V,E,F$, then $E=2F$, $\sum kn_k = 4F$, $\sum n_k = V$, where $n_k$ is the number of vertices with valence $k$, furthermore Euler formula holds, $V+F-E=\chi(\Sigma)$, where $\chi(\Sigma)$ is the Euler characteristic number of $\Sigma$.

\paragraph{Riemannian metric structure} If each face of $\mathcal{Q}$ is treated as the unit planar square, a flat Riemannian metric with cone singularities is induced, denoted as $\mathbf{g}_Q$. A vertex with $k$-valence has discrete curvature $(4-k)/2\pi$, the \emph{Gauss-Bonnet condition} shows the total curvature equals to the product of $2\pi$ and the Euler characteristic number $\chi(\Sigma)$. This implies the degree of the divisor equals
\begin{equation}
    \sum_{v\in Q}  (k(v)-4)  = -8\chi(\Sigma).
    \label{eqn:gauss_bonnet}
\end{equation}
The holonomy group induced by the metric $\mathbf{g}_Q$ on the surface with punctures at the singular vertices is the rotation group
\begin{equation}
\{e^{i\frac{\pi}{2}k},k\in \mathbb{Z}\}.
\label{eqn:holonomy}
\end{equation}
This is the so-called \emph{holonomy condition}. Furthermore, if we connect the horizontal and vertical edges of the quad-faces, we get geodesic loops. If we subdivide the quad-mesh infinitely many times, we obtain geodesic lamination, each leaf is a closed loop. This is called the \emph{finite geodesic lamination condition} \cite{CMAME_Quad_Mesh_I}.

\paragraph{Conformal Structure} In this work, we show that the quad-mesh $\mathcal{Q}$ induces a conformal structure, and can be treated as a Riemann surface $S_Q$. Furthermore, it induces a meromorphic quartic differential $\omega_Q$, whose trajectories are finite. Naturally the quad-mesh divisor $D_Q$ is equivalent to the divisor of $\omega_Q$. By Abel theorem, the divisor satisfies the Abel-Jacobi condition in Eqn.~\ref{eqn:abel_condition}. Inversely, if a point configuration $D$ satisfies the Abel-Jacobi condition, then there must be a meromorphic quartic differential $\omega$, whose divisor equals to $D$. If the trajectories of $\omega$ are finite, then $\omega$ induces a quad-mesh.

\paragraph{Comparison Between Structures}

The metric structure $\mathbf{g}_Q$ gives partial information about the divisor $D_Q$; whereas the conformal structure $S_Q$ and the meromorphic quartic differential $\omega_Q$ gives more thorough information about $D_Q$.

Given any divisor $D=\sum_k n_k p_k$, where $n_k$ is an integer, and the total number of $p_k$'s is finite. If $\sum n_k = -4\chi(\Sigma)$, then one can construct a flat Riemannian metric conformal to the original metric using Ricci flow \cite{}. The flat metric is with cone singularities at $p_k$, where the angle equals to $\pi(n_k+4)/2$. But this flat metric may not satisfies the holonomy condition in Eqn.~(\ref{eqn:holonomy}). If furthermore the divisor $D$ satisfies the Abel-Jacobi condition, then the obtained flat metric satisfies the holonomy condition.

We use a simple example to demonstrate the power of Abel-Jacobi condition. The following question is raised in \cite{Myles:2014:RFG:2601097.2601154}:
\begin{problem} Is there a quad-mesh on a closed torus, such that it has only two singularities, one valence $3$ vertex and one valence $5$ vertex, other vertices are regular (with valence $4$) ?
\label{prob:1}
\end{problem}
From heuristic experiments, it seems that such a quad-mesh does not exist. But it is difficult to find a rigorous argument: from topological point of view, the connectivity satisfies the Euler equation; from geometric point of view, there exists a flat Riemannian metric with the two cone singularities with curvature $\pi/2$ and $-\pi/2$ corresponding to the valence $3$ and valence $5$ vertices. But by Abel-Jacobi condition, we can show such kind of quad-mesh doesn't exist in corollary \ref{cor:genus_one}.

\subsection{Contributions}
This work opens a novel direction for quad-mesh generation based on Riemann surface theory:
\begin{enumerate}
    \item To the best of our knowledge, this is the first work that discovers the intrinsic connection between quadrilateral meshes and meromorphic quartic differentials (theorems \ref{thm:quad_differential} and \ref{thm:differential_quad}).
    \item This work gives the necessary and sufficient conditions for the singularity configuration, the Abel-Jacobi conditions for divisors \ref{thm:Abel_Jacobian_condition}.
    \item This work proposes to generate a quad-mesh by construct the corresponding meromorphic quartic differential with global algebraic representation.
\end{enumerate}

The work is organized as follows: section \ref{sec:previous_works} briefly review the most related works; section \ref{sec:theory} introduces the theoretic background, section \ref{sec:quad_differential} proves the equivalence between quad-meshes and meromorphic quartic differentials, Abel-Jacobi conditions; section \ref{sec:algorithm} explains the algorithm in details, and give simple examples to verify Abel-Jacobi conditions and construct meromorphic quartic differentials; finally, the work concludes in section \ref{sec:conclusion}.

\section{Previous Works}
\label{sec:previous_works}

There are many approaches for quadrilateral mesh generation, a thorough survey can be found in \cite{survey:Bommes2013Quad}. In the following, we only briefly mention the most related works.

\paragraph{Triangle Mesh Conversion} Quad-meshes can be generated by conversion from triangular meshes directly. The simplest way is to insert the bary-centers of faces and edges to obtain the initial quad-mesh, then perform Catmull-Clark subdivision. Alternatively, two original adjacent triangles can be fused into one quadrilateral to form a quad-mesh \cite{Gurung2011SQuad,Remacle2012Blossom,Marco2010Practical,Velho20014}. This approach can only produce unstructured quad-meshes, the quad-mesh quality is determined by the input triangle mesh.

\paragraph{Patch Based Approach}

This approach computes the coarsest level quadrilateral tessellation first, where the cutting graph is called the skeleton, then coarse mesh is subdivided to obtain finer level quad-meshes. The skeleton can be generated by clustering method, which merge neighboring triangles into a patch, including normal-based and center-based methods \cite{Boier2004Parameterization,Carr2006Rectangular}. Another method is to deform the input surface into a polycube shape, which is the union of cubes, then the faces of the polycube give the patches \cite{Xia2011Editable,Wang2008User,Lin2008Automatic,He2009A}. This approach can generated semi-regular quad-meshes.

\paragraph{Parameterization Based Approach}
Parameterization based approach computes the quadrilaterial tessellation in the parameter domain, or finds the skeleton from intrinsic geometric functions or differentials.
The spectral surface quadrangulation method \cite{Dong2006Spectral,Huang2008Spectral} produces the skeleton structure from the Morse-Smale complex of an eigenfunction of the Laplacian operator on the input mesh. Discrete harmonic forms \cite{Tong2006Designing}, periodic Global Parameterization \cite{Alliez2006Periodic} and Branched Coverings method \cite{K2010QuadCover} are all based on parameterization for quad mesh generation.

\paragraph{Voronoi Based Method} Centroidal Voronoi Tessellation (CVT) produces a surface cell decomposition with uniform cell sizes and shapes. The method in \cite{L2010Lp} generalizes the CVT from $L^2$ distance to general $L^p$ distance, when $p$ goes to infinity, the cells tens to be quadrilateral, this method allows for aligning the axes of the Voronoi cells with a predefined background tensor field. This method can only produce non-structured quad-mesh.

\paragraph{Cross field Based Approach}
One of the most popular approach is cross field guided quad-mesh generation. Each algorithm first choose a way to represent a cross, for example N-RoSy representation\cite{Palacios2007Rotational}, period jump technique\cite{Li2006Representing} and complex value representation\cite{Kowalski2013A}; then the algorithm usually generate a smooth cross field by energy minimization technique, such as discrete Dirichlet energy optimization\cite{JFH}. In the end, based on the obtained cross field, these approaches generate the quad meshes by using streamline tracing techniques\cite{RS:RPS:2014} or parameterization method \cite{survey:Bommes2013Quad}.The cross field guided quad mesh generation method can be very useful and flexible. However it is difficult to control the position of the singularities and the structures of the quad layout directly.

The work in \cite{landau} relates the Ginzberg-Landau theory with the cross field for genus zero surface case. This work further generalizes the work in \cite{landau} by relating Riemann surface theory with the cross fields for surfaces with arbitrary topologies. In theory, cross field gives the horizontal/vertical directions of a meromorphic quartic differential, but ignores the amplitude. Therefore, a meromorphic differential is a more precise representation of a quad-mesh.

\paragraph{Metric Based Approach} A quad-mesh induces a flat metric with cone singularities. Furthermore, the work in \cite{CMAME_Quad_Mesh_I} shows the holonomy group of the metric has special properties. Therefore the method in \cite{CMAME_Quad_Mesh_I} proposes to construct a flat metric using ricci flow algorithm with singularities at the given points, such that a quad-mesh can be induced when the holonomy conditions are met. The existence of the solution to the Ricci flow has theoretic guarantees. However the holonomy condition heavily depends on the singularity configuration. The work in \cite{CMAME_Quad_Mesh_I} didn't answer when the singularity configurations are appropriate for the holonomy condition.

In contrast, the current work gives the sufficient and necessary condition for the singularity configuration in order to satisfy the holonomy requirements: the Abel-Jacobi condition.

\paragraph{Holomorphic Differential Approach}

In \cite{Zheng2019} and \cite{lei2017generalized}, the holomorphic quadratic differential is utilized to generate quad-meshes and hex-meshes, this approach produces quad-meshes with least singularities and highest smoothness. However, this approach can not model singularities with odd topological valences, which greatly prevents the method for general applications in practice.

The current work is a direct generalization of this approach, by generalizing holomorphic quadratic differentials to more general meromorphic quartic differentials. This conquers the difficulties raised in the holomorphic differential method, and covers all possible quad-meshes.

Comparing with all existing approaches, current work shows the equivalence between quad-meshes and meromorphic quartic differentials, and gives the Abel-Jacobi condition for singularities. This picture is most general and complete, generalizes most existing approaches including cross fields, Strebel differential and metric based methods.

\section{Theoretic Background}
\label{sec:theory}
This section briefly review the basic concepts and theorems in Riemann surface theory, details can be found in \cite{}.

\subsection{Riemann Surface}

\begin{definition}[Topological Manifold] Suppose $\Sigma$ is a topological space, $\{U_\alpha\}$ is a family of open sets covering the space, $\Sigma\subset \bigcup_\alpha U_\alpha$. For each open set $U_\alpha$, there exists a homeomorphism $\varphi_\alpha:U_\alpha \to \mathbb{R}^n$, the pair $(U_\alpha,\varphi_\alpha)$ is called a local  chart. The collection of local charts form the atlas of $M$, $\mathcal{A}=\{(U_\alpha,\varphi_\alpha)\}$. For any pair of open sets, $U_\alpha$ and $U_\beta$, if $U_\alpha\cap U_\beta \neq \emptyset$, the transition map is given by $\varphi_{\alpha\beta}:\varphi_\alpha(U_\alpha\cap U_\beta)\to \varphi_\beta(U_\alpha\cap U_\beta)$, $\varphi_{\alpha\beta} = \varphi_\beta\circ \varphi_\alpha^{-1}$.
Then $\Sigma$ is called a closed $n$-dimesnional manifold.
\end{definition}
Two dimensional manifolds are called surfaces.

\begin{definition}[Conformal Atlas] Suppose $S$ is a two dimensional topological manifold, equipped with an atlas $\mathcal{A}=\{(U_\alpha,\varphi_\alpha)\}$, every local chart are complex coordinates $\varphi_\alpha:U_\alpha\to \mathbb{C}$, denoted as $z_\alpha$, and every transition map is biholomorphic,
\[
    \varphi_{\alpha\beta}:\varphi_\alpha(U_\alpha\cap U_\beta)\to \varphi_\beta(U_\alpha\cap U_\beta), \quad z_\alpha \mapsto z_\beta,
\]
then the atlas is called a conformal atlas.
\end{definition}

\begin{definition}[Riemann Surface] A topological surface with a conformal atlas is called a Riemann surface.
\end{definition}

\begin{definition}[biholomorphic map] Suppose $f:(S,\{(U_\alpha, \varphi_\alpha\})$ to $(T,\{(V_\beta,\psi_\beta)\})$ is a map between two Riemann surfaces, if every local representation
\[
    \psi_\beta \circ f \circ \varphi_\alpha^{-1}: \varphi_\alpha(U_\alpha)\to \psi_\beta(V_\beta)
\]
is biholomorphic, then $f$ is called a biholomorphic map between Riemann surfaces, namely a conformal map.
\end{definition}

Suppose $(S,\mathbf{g})$ is an oriented surface with a Riemannian metric $\mathbf{g}$. For each point $p\in \Sigma$, we can find a neighborhood $U(p)$, inside $U(p)$ the \emph{isothermal coordinates} $(u,v)$ can be constructed, such that $\mathbf{g}=e^{2\lambda(u,v)}(du^2+dv^2)$. The atlas formed by all the isothermal coordinates is a conformal atlas, therefore the surface $(S,\mathbf{g})$ is a Riemann surface:

\begin{theorem}
All oriented surfaces with Riemannian metrics are Riemann surfaces.
\end{theorem}

Two metrics $\mathbf{g}_1,\mathbf{g}_2$ on a surface $\Sigma$ are \emph{conformal equivalent} to each other, if there is a scalar function $\lambda:S\to\mathbb{R}$, such that $\mathbf{g}_1 = e^{2\lambda}\mathbf{g}_2$.

\subsection{Meromorphic Functions and Differentials}

\begin{definition}[Holomorphic Function] Suppose $f:\mathbb{C}\to\mathbb{C}$ is a complex function, $(x,y)\mapsto (u(x,y),v(x,y))$, if the function satisfies the Cauchy-Riemann equation
\[
    \frac{\partial u}{\partial x} = \frac{\partial v}{\partial y},
    \quad\frac{\partial u}{\partial y} = -\frac{\partial v}{\partial x},
\]
then $f$ is called a holomorphic function. If $f$ is invertible, furthermore $f^{-1}$ is also holomorphic, then $f$ is called biholomorphic.
\end{definition}

\begin{definition}[Meromorphic Function] Suppose $f:\mathbb{C}\to\mathbb{C}\cup \{\infty\}$ is a complex function, $f(z)=p(z)/q(z)$, where $p(z)$ and $q(z)$ are holomorphic functions, then $f(z)$ is called a meromorphic function.
\end{definition}

\begin{definition}[Laurent Series] The Laurent series of a meromorphic function about a point $z_0$ is given by
\[
    f(z) = \sum_{n=k}^\infty a_n(z-z_0)^n,
\]
the series $\sum_{k\ge 0}^\infty a_n(z-z_0)^n$ is called the analytic part of the Laurent series; the series $\sum_{n<0} a_n(z-z_0)^n$ is called the principal part of the Laurent series. $a_{-1}$ is called the residue of $f$ at $z_0$.
\end{definition}

\begin{definition}[Zeros and Poles]
Given a meromorphic function $f(z)$, if its Laurent series has the form
\[
    f(z) = \sum_{n=k}^\infty a_n(z-z_0)^n,
\]
if $k>0$, then $z_0$ is a zero point of order $k$; if $k<0$, then $z_0$ is called a finite pole of $f(z)$ of order $k$; if $k=0$, then $z_0$ is called a regular point. The order of a zero or a pole at the point $p$ of $f$ is denoted as $\nu_p(f)$.
\end{definition}

The concepts of holomorphic and meromorphic functions can be generalized to Riemann surfaces.

\begin{definition}[Meromorphic Function on Riemann Surface] Suppose a Riemann surface $(S,\{(U_\alpha, z_\alpha)\})$ is given. A complex function is defined on the surface $f:S\to \mathbb{C}\cup \{\infty\}$. If on each local chart $(U_\alpha, z_\alpha)$, the local representation of the functions $f\circ \varphi_\alpha^{-1}:\mathbb{C}\to \mathbb{C}\cup\{\infty\}$ is meromorphic, then $f$ is called a meromorphic function defined on $S$.
\end{definition}
A memromorphic function can be treated as a holomorphic map from the Riemann surface to the unit sphere.

\begin{definition}[Meromorphic Differential] Given a Riemann surface $(S,\{z_\alpha\})$, $\omega$ is a meromorphic differential of order $n$, if it has local representation,
\[
    \omega = f_\alpha(z_\alpha) (dz_\alpha)^n,
\]
where $f_\alpha(z_\alpha)$ is a meromorphic function, $n$ is an integer; if $f_\alpha(z_\alpha)$ is a holomorphic function, then $\omega$ is called a holomorphic differential of order $n$.
\end{definition}
A holomorphic differential of order $1$ is called a \emph{holomrphic 1-form};
A holomorphic differential of order $2$ is called a \emph{holomrphic quadratic differential}; A meromorphic differential of order $4$ is called a \emph{meromorphic quartic differential}.


\begin{figure}[h!]
\centering
\begin{tabular}{cc}
\includegraphics[height=0.45\textwidth]{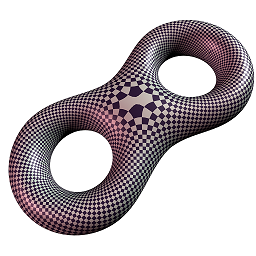}
&\includegraphics[height=0.45\textwidth]{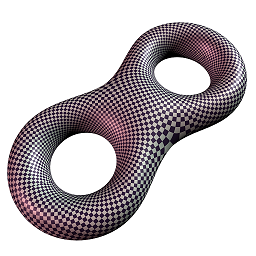}\\
\end{tabular}
\caption{Zeros on holomorphic 1-forms on a genus two surface.}
\label{fig:genus_two}
\end{figure}
\begin{definition}[Zeros and Poles of Meromorphic Differentials] Given a Riemann surface $(S,\{z_\alpha\})$, $\omega$ is a meromorphic differential with local representation,
\[
    \omega = f_\alpha(z_\alpha) (dz_\alpha)^n.
\]
If $z_\alpha$ is a pole (or a zero) of $f_\alpha$ with order $k$, then $z_\alpha$ is called a pole (or a zero) of the meromorphic differential $\omega$ of order $k$.
\end{definition}

We use $Sing_\omega$ to denote the singularity set of $\omega$. Locally near a regular point $p$, the differential $\omega=f(z)(dz)^n$ can be represented as the $n$-th power of a 1-form $h(z)dz$ where $h^n(z)=f(z)$ and thus $h(z)=\sqrt[n]{f(z)}$ coincides with one of $n$ possible branches of the $n$-th root. We call this $n$-valued 1-form the \emph{$n$-th roots} of $\omega$, which is a globally well-defined multi-valued meromorphic 1-form on $S$.

\begin{figure}[h!]
\centering
\begin{tabular}{cc}
\includegraphics[height=0.4\textwidth]{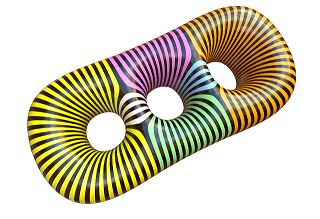}
&\includegraphics[height=0.35\textwidth]{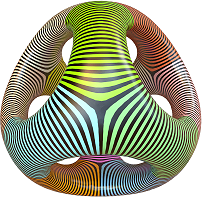}\\
\end{tabular}
\caption{Horizontal trajectories of holomorphic quadratic differentials.}
\label{fig:foliations}
\end{figure}

\begin{definition}[Trajectories of meromorphic differentials]Given a meromorphic $n$-differential $\omega$ on $S$ we define $n$ distinct line fields on $S\setminus Sing_\omega$ as follows. At each non-singular point $z$ there are exactly $n$ distinguished directions $dz$ at which $\omega=f(z)(dz)^n$ attains real values. Integral curves of these line fields are called trajectories of $\omega$.
\end{definition}

Suppose $\omega$ is a meromorphic quadratic differential, $dz$ is a \emph{horizontal (vertical) direction} if $f(z)(dz)^2>0$ ($f(z)(dz)^2<0$). Integral curves of horizontal direction are called \emph{horizontal (vertical) trajectories}.

\begin{definition}[Strebel Differential] A meromorphic quadratic differential is called a Strebel differential, if all its horizontal trajectories are finite.
\end{definition}

Fig.~\ref{fig:foliations} shows the horizontal trajectories of Strebel differentials. Note that the vertical trajectories of a Strebl differential may not necessarily be finite.

\subsection{Divisor}

\begin{definition}[Divisor] The Abelian group freely generated by points on a Riemann surface is called the divisor group, every element is called a divisor, which has the form
\[
    D = \sum_p n_p p,
\]
where only a finite number of $n_p$'s are non-zeros. The degree of a divisor is defined as $deg(D)=\sum_p n_p$. Suppose $D_1 = \sum_p n_p p$, $D_2 = \sum_p m_p p$, then $D_1\pm D_2 = \sum_p(n_p\pm m_p)p$; $D_1\le D_2$ if and only if for all $p$, $n_p \le m_p$.
\end{definition}

\begin{definition}[Meromorphic Function Divisor] Given a meromorphic funciton $f$ defined on a Riemann surface $S$, its divisor is defined as
\[
    (f) = \sum_p \nu_p(f)p.
\]
\end{definition}
The divisor of a meromorphic differential $\omega$ is defined in the similar way.

\begin{definition}[Meromorphic Differential Divisor] Suppose $\omega$ is a meromorphic differential on a Riemann surface $S$, suppose $p\in S$ is a point on $S$, we define the order of $\omega$ at $p$ as
\[
    \nu_p(\omega) = \nu_p(f_p),
\]
where $f_p$ is the local representation of $\omega$ in a neighborhood of $p$, $\omega= f_p (dz_p)^n$.
\end{definition}

\begin{definition}[Principle Divisor]
The divisors of meromorphic functions are called principle divisors.
\end{definition}
All principle divisors are of degree zeros. Suppose $\omega$ is a meromorphic 1-form, then $deg((\omega))=2g-2$, where $g$ is the genus.

\begin{definition}[Equivalent Divisors] Two divisors are equivalent, if their difference is a principle divisor.
\end{definition}

\subsection{Abel-Jacobian Theorem}

\begin{figure}[h!]
\centering
\begin{tabular}{cc}
\includegraphics[width=0.6\textwidth]{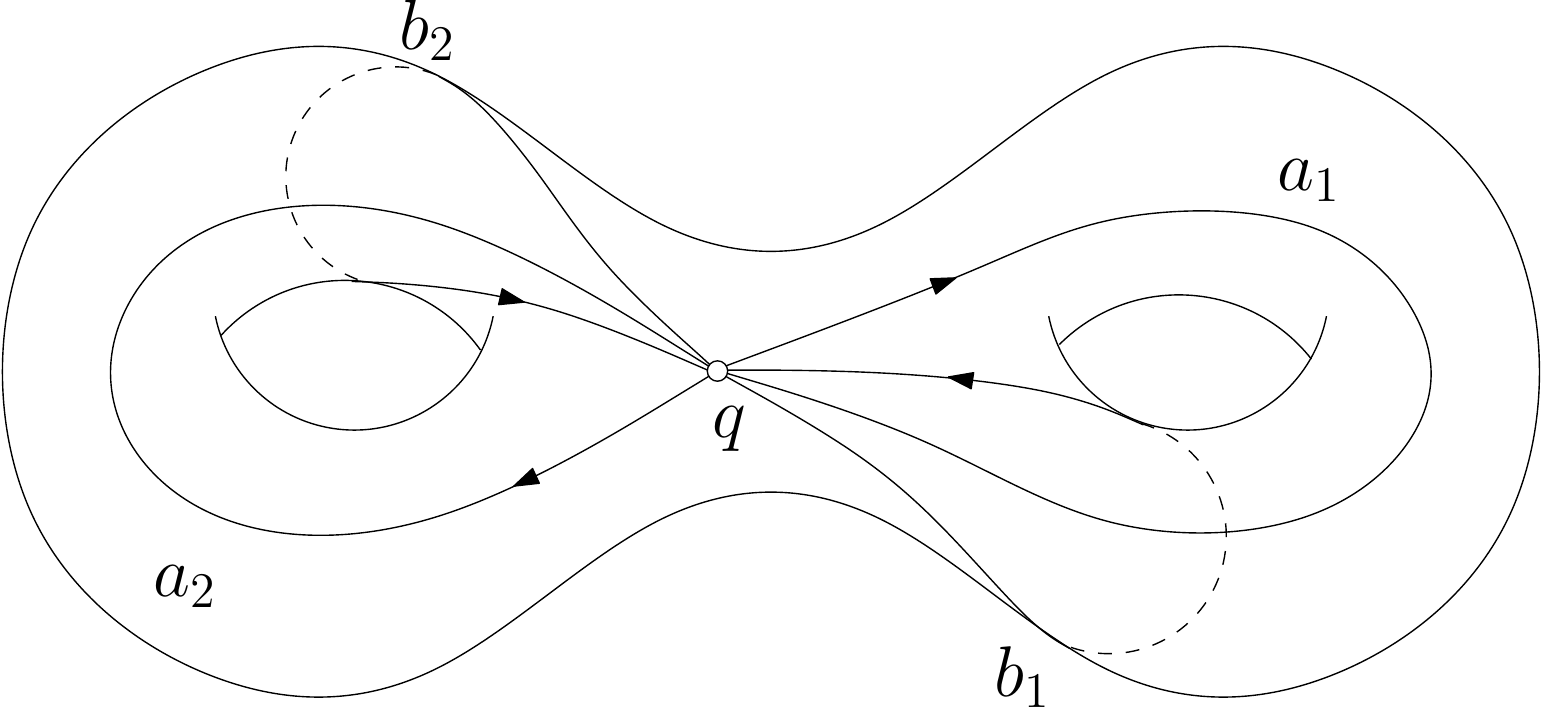}&
\includegraphics[width=0.32\textwidth]{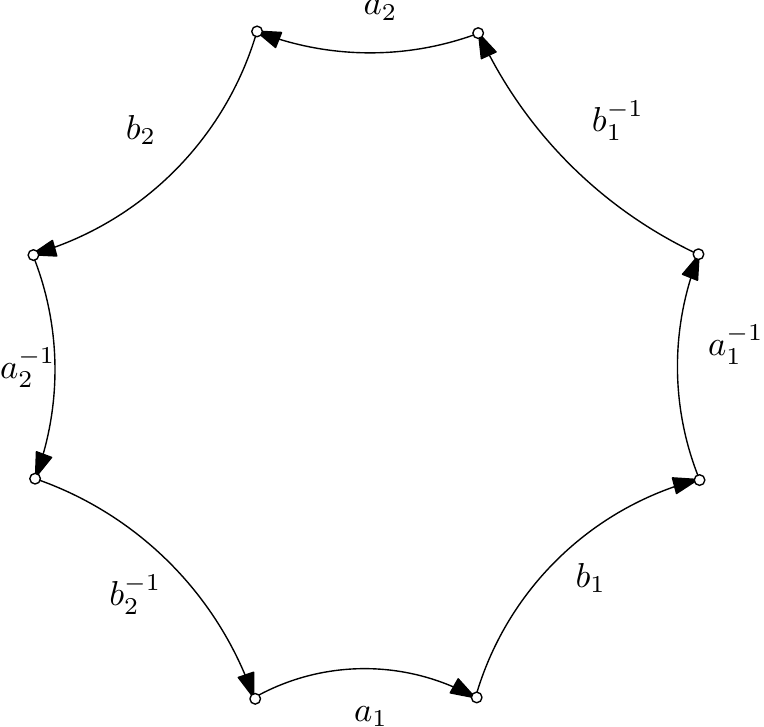}\\
a) canonical basis & b) fundamental domain\\
\end{tabular}
\caption{Canonical fundamental group basis.}
\label{fig:fundamental group basis}
\end{figure}

Suppose $\{a_1,b_1,\dots,a_g,b_g\}$ is  a set of canonical basis for the homology group $H_1(S,\mathbb{Z})$, each $a_i$ and $b_i$ represent the curves around the inner and outer circumferences of the $i$th handle. The surface is sliced along the homology group basis to obtain a fundamental domain,  as shown in Fig.~\ref{fig:fundamental group basis}.

Let $\{\omega_1,\omega_2,\dots, \omega_g\}$ be a normalized basis of $\Omega^1$, the linear space of all holomorphic 1-forms over $\mathbb{C}$. The choice of basis is dependent on the homology basis chosen above; the normalization signifies that
\[
    \int_{a_i} \omega_j = \delta_{ij},\quad i,j = 1,2,\dots,g.
\]
For each curve $\gamma$ in the homology group, we can associate a vector $\lambda_\gamma$ in $\mathbb{C}^g$ by integrating each of the $g$ 1-forms over $\gamma$,
\[
    \lambda_\gamma = \left(\int_\gamma \omega_1,\int_\gamma \omega_2,\dots, \int_\gamma \omega_g \right)
\]
The \emph{period matrix} of the Riemann surface $S$ is given by:
\[
    \left(\lambda_{a_1},\lambda_{a_2},\cdots,\lambda_{a_g};\lambda_{b_1}, \lambda_{b_2},\cdots, \lambda_{b_g}\right).
\]
We define a $2g$-real-dimensional lattice $\Lambda$ in $\mathbb{C}^g$,
\[
    \Gamma = \left\{  \sum_{i=1}^g \alpha_i~\lambda_{a_i} + \sum_{j=1}^g \beta_j~ \lambda_{b_j}, \quad \alpha_i, \beta_j\in \mathbb{Z} \right\}
\]
\begin{definition}[Jacobian Variety]
The Jacobian Variety of the Riemann surface $S$, denoted $J(S)$, is the compact quotient $\mathbb{C}^g/\Lambda$.
\end{definition}

\begin{definition}[Abel-Jacobi Map] Fix a base point $p_0\in S$. The Abel-Jacobi map is a map $\mu:S\to J(S)$. For every point $p\in S$, choose a curve $\gamma$ from $p_0$ to $p$ inside the fundamental domain; the Abel-Jacobi map $\mu$ is defined as follows:
\[
    \mu(p) = \left(\int_\gamma \omega_1, \int_\gamma \omega_2, \dots, \int_\gamma \omega_g \right) ~~\mod \Gamma,
\]
where the integrals are all along $\gamma$.
\end{definition}
It can be shown $\mu(p)$ is well-defined,  that the choice of curve $\gamma$ doesn't not affect the value of $\mu(p)$. Given a divisor $D=\sum_k n_k p_k$, the Abel-Jacobi map is defined as
\[
    \mu(D) = \sum_k n_k \mu(p_k).
\]
\begin{theorem}[Abel Theorem] Given a compact Riemann surface $S$ of genus $g>0$, and a degree zero divisor $D$, $D$ is a principle divisor if and only if
\begin{equation}
\mu(D)=0 \quad \text{in}\quad J(S).
\end{equation}
\label{thm:abel_theorem}
\end{theorem}

\section{Quad-Meshes and Meromorphic Quartic Forms}
\label{sec:quad_differential}

In this section, we show the intrinsic relation between a quad-mesh and a meromorphic quartic differential, this gives the Abel-Jacobi condition for the singularity configuration of any quad-mesh.

\subsection{Quadrilateral Mesh}

\begin{definition}[Quadrilateral Mesh] Suppose $\Sigma$ is a topological surface, $\mathcal{Q}$ is a cell partition of $\Sigma$, if all cells of $\mathcal{Q}$ are topological quadrilaterals, then we say $(\Sigma,\mathcal{Q})$ is a quadrilateral mesh.
\end{definition}


On a quad-mesh, the \emph{topological valence} of a vertex is the number of faces adjacent to the vertex.
\begin{definition}[Singularity] Suppose $(S,\mathcal{Q})$ is a closed quadrilateral mesh. For each vertex $v$ with topological valence $k$. If $k$ is $4$, then we call $v$ a \emph{regular vertex}, otherwise a \emph{$k$-singular vertex}.
\end{definition}

\begin{definition}[Quad-mesh Metric] Suppose $(\Sigma,\mathcal{Q})$ is a closed quadrilateral mesh. Each face is treated as a unit planar square, then $\mathcal{Q}$ induces a flat Riemannian metric $\mathbf{g}_Q$ with cone singularities at singularities.
\end{definition}

\begin{definition}[Discrete Curvature] Suppose $(\Sigma,\mathcal{Q})$ is a closed quadrilateral mesh. Under the metric $\mathbf{g}_Q$, the discrete curvature at a valence $k$ vertex is
\[
    \frac{\pi}{2} (4-k).
\]
\end{definition}

\begin{theorem}[Gauss-Bonnet]Suppose $(\Sigma,\mathcal{Q})$ is a closed quadrilateral mesh of genus $g$. Under the metric $\mathbf{g}_Q$, the total curvature is
\begin{equation}
    \sum_{v\in Q}  \frac{\pi}{2} (4-k(v))= 2\pi \chi(\Sigma).
    \label{eqn:gauss_bonnet}
\end{equation}
\end{theorem}

\begin{figure}[h!]
\centering
\begin{tabular}{cc}
\includegraphics[height=0.45\textwidth]{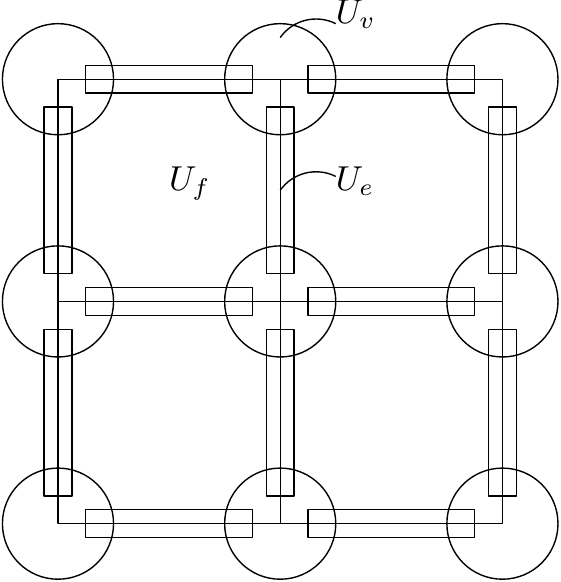}&
\includegraphics[height=0.35\textwidth]{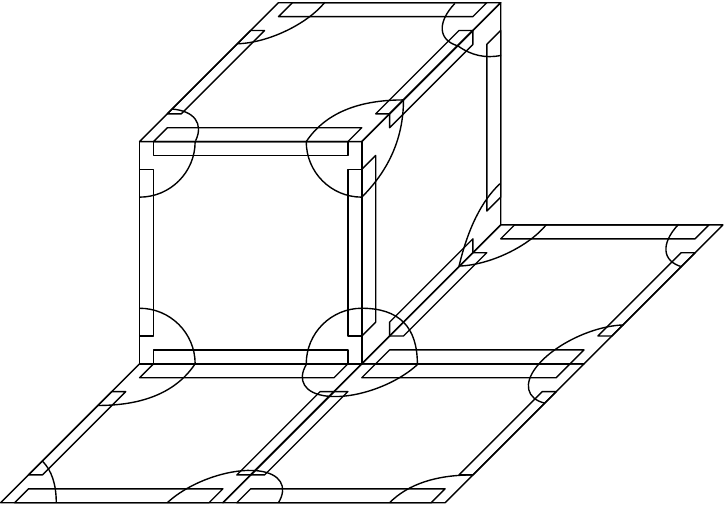}\\
(a) conformal atlas &(b) singularities\\
\end{tabular}
\caption{A quad-mesh induces a conformal atlas, such that the surface becomes a Riemann surface.}
\label{fig:quad_mesh_conformal_atlas}
\end{figure}

\subsection{Quad-mesh and Meromorphic Quartic Differential}

\begin{theorem}[Quad-Mesh and Riemann Surface]
Suppose $(\Sigma,\mathcal{Q})$ is a closed quadrilateral mesh, then the quad-mesh $\mathcal{Q}$ induces a conformal atlas $\mathcal{A}$, such that $(\Sigma,\mathcal{A})$ form a Riemann surface, denoted as $S_Q$.
\end{theorem}

\begin{proof} As shown in Fig.~\ref{fig:quad_mesh_conformal_atlas}, we treat each face as the unit Euclidean planar square, this assignes a flat Riemannian metric to the surface, such that the curvature is zero everywhere except at the singularities. At the singularity with valence $k$, the cone angle is $k\pi/2$.

First, we construct the open covering of the surface. The interior of each face $f$ is one open set $U_f$; each edge $e$ is covered by the interior of a rectangle $U_e$; each vertex $v$ is covered by a disk, $U_v$. It is obvious that
\[
    \Sigma \subset \bigcup_f U_f \bigcup_e U_e \bigcup_v U_v.
\]
Second, we establish the local complex parameter of each open set. For each face $U_f$, we choose the center of the face as the origin, the real and imaginary axises are parallel to the edges of the face, the local parameter is denoted as $z_f$; for each edge $U_e$, we choose the center as the origin, the real axis is parallel to the edge direction, the local parameter is $z_e$; for each regular vertex $U_v$, we choose the center of the disk as the origin, one edge as the real axis, the local parameter is $z_v$; for a singular vertex $U_v$, we first isometrically flatten $U_v$ on the complex plane, to obtain the local parameter $w_v$, the we use the complex power function to shrink it as
\[
    w_v^{\frac{4}{k}}\mapsto z_v,
\]
where $k$ is valence of the singular vertex.

Now we examine all the coordinate transition maps. We denote a subgroup of the planar rigid motion $G$ generated by
\[
    z \mapsto e^{i\frac{\pi}{2}}z, ~~~z\mapsto z + \frac{1}{2}, ~~~z\mapsto z +  \frac{i}{2},
\]
Then any transition map $\varphi_{fe}$ between a face and an adjacent edge $z_f\to z_e$ is an element in $G$, $\varphi_{fe}\in G$; and transition map $\varphi_{ev}$ between an edge and one of its end vertex $z_e\to z_v$ is an element in $G$, $\varphi_{ev}\in G$.

The transition between a singulary vertex and its neihboring edge or face is more complicated. For example, the transition from a face to its neighboring singular vertex is given by
\begin{equation}
    z_v = \left(e^{i\frac{n\pi}{2}}z_f + \frac{1}{2}(\pm 1 \pm i)\right)^\frac{4}{k}
    \label{eqn:vertex_face_transition}
\end{equation}

where $m,n\in \mathbb{Z}$, $k$ is the valence of $v$.

Hence all the transition maps are biholomorphic. $\mathcal{A}=\{(U_f,z_f), (U_e,z_e), (U_v, z_v)\}$ form a conformal atlas. $(\Sigma,\mathcal{A})$ is a Riemann surface.
\end{proof}

\begin{theorem}[Qaud-Mesh to Meromrophic Quartic Differential]
Suppose $(\Sigma,\mathcal{Q})$ is a closed quadrilateral mesh, then the quad-mesh $\mathcal{Q}$ induces a quartic differential $\omega_Q$ on $S_Q$. The valence-$k$ singular vertices correspond to poles or zeros of order $k-4$. Furthermore, the trajectories of $\omega_Q$ are finite.
\label{thm:quad_differential}
\end{theorem}

\begin{proof}
We construct a holomorphic 1-form on each face $U_f$, $dz_f$. Because the orientations of all faces are chosen individually, $dz_f$ is not globally defined. Then we define the holomorphic quadric form
\[
    \omega = (dz_f)^4,
\]
the $\omega$ is globally defined. Suppose a face $f$ and an edge $e$ are adjacent, then
\[
    dz_f = e^{i\frac{n\pi}{2}}dz_e, ~n\in \mathbb{Z},
\]
then $ (dz_f)^4 = (dz_e)^4$. Suppose an edge $e$ has an regular edge vertex $v$
\[
    dz_e = e^{i\frac{n\pi}{2}} dz_v,~n\in \mathbb{Z},
\]
therefore $(dz_e)^4 = (dz_v)^4$. Suppose a face $f$ has a regular vertex $v$, according to Eqn.~\ref{eqn:vertex_face_transition}, where $k=4$,
\[
    dz_v = e^{i\frac{n\pi}{2}} dz_f,
\]
hence $(dz_v)^4 = (dz_f)^4$.

Suppose $v$ is a singular vertex with valence $k \neq 4$, by Eqn.~\ref{eqn:vertex_face_transition}, we obtain
\[
    \frac{k}{4}z_v^{\frac{k-4}{4}} dz_v = e^{i\frac{n\pi}{2}} dz_f,
\]
therefore
\begin{equation}
    \left(\frac{k}{4}\right)^4 z_v^{k-4} (dz_v)^4 = (dz_f)^4 = \omega.
    \label{eqn:zero_pole}
\end{equation}
Hence, a valence $3$ singular vertex corresponds to a simple pole $\frac{c}{z_v} dz_v^4$, a valence $5$ singular vertex becomes a simple zero $c z_v dz_v^4$.

Therefore, the meromorphic quartic differential $\omega$ is globally defined, the valence-$k$ singular vertices correspond to poles or zeros with order $k-4$. The finiteness of the trajectories of $\omega$ is obvious.
\end{proof}

\begin{theorem}[Quartic Differential to Quad-Mesh]
Suppose $(\Sigma,\mathcal{A})$ is a Riemann surface, $\omega$ is a meromorphic quartic differential with finite trajectories, then $\omega$ induces a quadrilateral mesh $\mathcal{Q}$, such that the poles or zeros with order $k$ of $\omega$ corresponds to the singular vertices of $\mathcal{Q}$ with valence $k+4$.
\label{thm:differential_quad}
\end{theorem}

\begin{proof}
The horizontal and vertical trajectories of $\omega$ partition the surface into rectangles. Given a pole with local representation
\[
    \omega = c_1 z^k dz^4,
\]
the holomorphic 1-form is given by $\sqrt[4]{\omega}$ with local representation \[
    c_2 z^{\frac{k}{4}}dz = c_2 d( z^\frac{k+4}{4}),
\]
the integration of  $\sqrt[4]{\omega}$ maps the $2\pi$ angle to $\frac{k+4}{2}\pi$, therefore the valence of the singular vertex equals to $k+4$.
\end{proof}

Fig.~\ref{fig:genus_two} shows the quad-meshes induced by holomorphic differentials.

\subsection{Abel-Jacobi condition}

\begin{definition}[Quad-Mesh Divisor]
Suppose $\mathcal{Q}$ is a closed quadrilateral mesh, $\omega_Q$ is the induced meromorphic quadric form. Then the quad-mesh induces a divisor $D_Q=(\omega_Q)$:
\[
    D_Q = (\omega_Q) = \sum_{v\in \mathcal{Q}} (k(v)-4) v,
\]
where $k(v)$ is valence of the vertex $v$.
\end{definition}

\begin{theorem}Suppose $\mathcal{Q}$ is a closed quadrilateral mesh of genus $g$, the degree of the induced divisor $D_Q$ is
\[
    \deg(D_Q) = 8g-8.
\]
\end{theorem}
\begin{proof}Directly induced by the Gauss-Bonnet condition  of $\mathbf{g}_Q$ in Eqn.~\ref{eqn:gauss_bonnet}.
\end{proof}
\begin{figure}[h!]
\centering
\begin{tabular}{rl}
\includegraphics[height=0.45\textwidth]{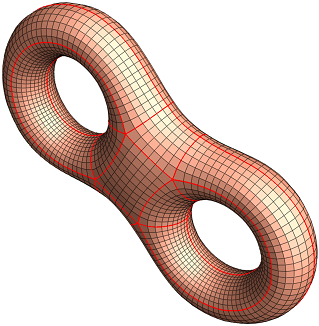}&
\includegraphics[height=0.45\textwidth]{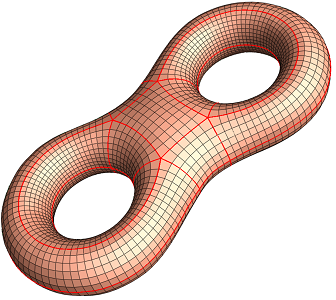}\\
 quad-mesh induces by a &
holomorphic quartic differential\\
\end{tabular}
\caption{Quadrilateral meshes induced by holomorphic differentials.}
\label{fig:genus_two}
\end{figure}

\begin{theorem}[Quad-mesh Abel-Jacobi condition]Suppose $\mathcal{Q}$ is a closed quadrilateral mesh, $S_Q$ is the induced Riemann surface, $D_Q$ is the induced divisor. Assume $\omega_0$ is an arbitrary holomorphic 1-form on $S_Q$, then
\begin{equation}
        \mu(D_Q - 4(\omega_0)) = 0
        \label{eqn:abel_condition}
\end{equation}
in the Jacobian $J(S_Q)$.
\label{thm:Abel_Jacobian_condition}
\end{theorem}

\begin{proof}
Suppose the meromorphic differential induced by $\mathcal{Q}$ is $\omega_Q$. The ratio $f=\omega_Q/\omega_0^4$ is a meromorphic function, therefore according to Abel theorem $\mu((f)) = 0$ in $J(S_Q)$,
\[
    \mu((f))=\mu((\omega_Q/\omega_0^4)) = \mu( (\omega_Q) - (\omega_0^4)) = \mu(D_Q - 4(\omega_0)) = 0~~~\mod \Gamma.
\]
in $J(S_Q)$.
\end{proof}

Namely, the divisors of all meromorphic differentials are equivalent, and are mapped to the same point in $J(S_Q)$ by the Abel-Jacobi map.

\begin{figure}[h!]
\centering
\begin{tabular}{cc}
\includegraphics[height=0.45\textwidth]{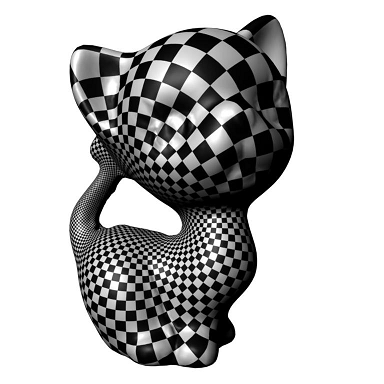}&
\includegraphics[height=0.45\textwidth]{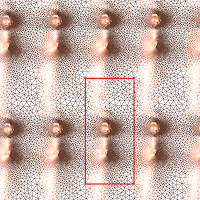}\\
(a) a holomorphic 1-form $\omega_0$ &(b) a fundamental domain $\Omega$\\
\end{tabular}
\caption{The proof of corollary ~(\ref{cor:genus_one}) based on Abel-Jacobi condition theorem ~(\ref{thm:Abel_Jacobian_condition}). No genus one closed quad-mesh has only one valence $3$ and one valence $5$ singular vertices.}
\label{fig:genus_one}
\end{figure}

\begin{corollary}Suppose $\mathcal{Q}$ is a genus one closed quadrilateral mesh. If $\mathcal{Q}$ has only one valence $3$ and one valence $5$ singular vertices, and no other singular vertex, then $\mathcal{Q}$ does not exist.
\label{cor:genus_one}
\end{corollary}

\begin{proof}
Suppose $p$ is the valence $3$ singularity, $q$ is the valence $5$ singularity. Then $p$ is the simple pole of $\omega_Q$, $q$ is the simple zero of $\omega_Q$. Suppose $\omega_0$ is the canonical holomorphic 1-form on the Riemann surface $S_Q$. $\{a,b\}$ is a set of canoincal homology group basis, $\Omega$ is a fundamental domain. Choose a base point $p_0\in \Omega$ and paths $\gamma_p,\gamma_q\subset \Omega$, connecting the base point to the pole and the zero. According to Abel-Jacobi condition in Eqn.~\ref{eqn:abel_condition},
\[
    \mu((p-q))=\int_{\gamma_p} \omega_0 - \int_{\gamma_q} \omega_0 = 0,
\]
therefore the pole $p$ and the zero $q$ coincide, name the valence $3$ and valence $5$ vertices coincide. Contradiction, hence such kind of $\mathcal{Q}$ does not exist.
\end{proof}

\begin{theorem} [Inverse theorem of Abel-Jacobi Condition]Suppose $S$ is a compact Riemann surface of genus $g$, $\omega_0$ is a holomorphic one-form, $D$ is a divisor, satisfying
\begin{equation}
    \mu(D-4(\omega_0))=0
    \label{eqn:quad_mesh_divisor}
\end{equation}
in $J(S)$, then there exists a meromorphic quartic differential $\omega$, such that $(\omega) = D$.
\label{thm:inverse}
\end{theorem}

\begin{proof}
According to Abel theorem, $\mu(D-4(\omega_0))=0$ in $J(S)$ implies there exists a meromorphic function $f$, such that $(f)=D-4(\omega_0)$, then let
\[
    \omega = f \cdot \omega_0^4,
\]
then $(\omega) = (f)+(\omega_0^4) = D$ is the desired meromorphic quartic differential.
\end{proof}

\begin{corollary}Suppose $S$ is a compact Riemann surface of genus $g$, $\omega_0$ is a holomorphic one-form, $D$ is a divisor, satisfying condition in Eqn.~\ref{eqn:quad_mesh_divisor} in $J(S)$, then there exists a meromorphic quartic differential $\omega$, such that $(\omega) = D$. Furthermore, if the trajectories of $\omega$ are finite, then $\omega$ induces a quadrilateral mesh $\mathcal{Q}$.
\end{corollary}
\begin{proof} The surface can be sliced along the trajectories of $\omega$, especially those through the zeros and poles. This produces a quadrilateral mesh.
\end{proof}

\section{Computational Algorithms}

This section explain the computational algorithms to verify Abel-Jacobi condition and construct meromorphic quartic differentials in details.

\subsection{Abel-Jacobi Condition Verification}

Given a closed triangle mesh, and a divisor, we would like to verify if the divisor satisfies the Abel-Jacobi condition. This involves the computation of the canonical basis for the homology group and holomorphic differential group. The canonical homology group basis is carried out using the geometry-aware handle loop and tunnel loop algorithm in \cite{Dey:2008:CGH:1360612.1360644}.
the holomorphic differential basis is computed using the algorithm described in \cite{gu2003global}.
The computation of the period matrix and Abel-Jacobi map are straight forward.  The details of the algorithm is described in Alg.~\ref{alg:Abel_Jacobi_condition}.

\begin{algorithm}[h]
\KwIn{Closed Surface $S$ of genus $g$; Divisor $D$ }
\KwOut{Whether $D$ satisfies Abel Condition}
\nl if $\deg(D) \neq 8g-8$ then return false;\\
\nl Compute the canonical homology group generators $\{a_1,\dots,a_g;b_1,\dots,b_g\}$;\\
\nl Compute the dual holomorphic 1-form basis $\{\omega_1,\cdots,\omega_g\}$;\\
\nl Compute the period matrix of $S$;\\
\nl Compute the Abel-Jacobi map $\mu(D)$;\\
\nl Solve equation group
\[
    (\text{Img}~B) \beta = \text{Img} \mu(D),
\]
to obtain $\beta=(\beta_1,\beta_2,\dots, \beta_g)^T$.\\
\nl Solve equation group
\[
    A\alpha = \mu(D) - B\beta
\]
to obtain $\alpha=(\alpha_1,\alpha_2,\dots, \alpha_g)^T$.\\
\nl if all $\alpha_i$'s and $\beta_j$'s are integer then return ture; otherwise return false.
    \caption{{\bf Abel Condition Verification} \label{alg:Abel_Jacobi_condition}}
\end{algorithm}

\begin{figure*}[h]
\centering
\begin{tabular}{cc}
	\includegraphics[height=0.33\textwidth]{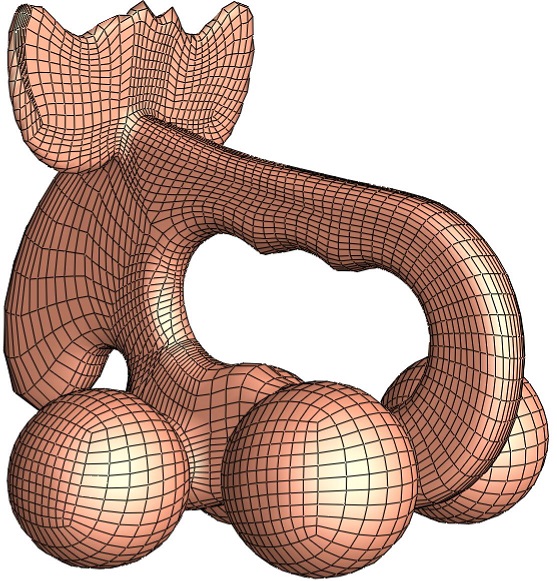} &
	\includegraphics[height=0.33\textwidth]{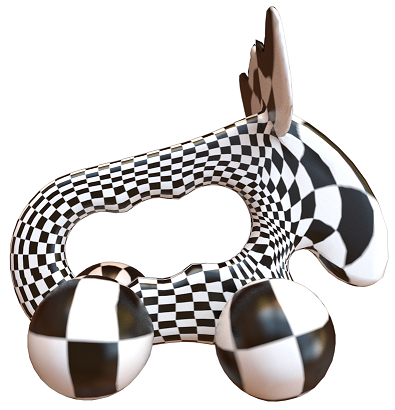} \\
	(a) original quad mesh & (b) holo 1 form  \\
\end{tabular}
\caption{The genus 1 elk model.}
	\label{fig:elk}
\end{figure*}

\paragraph{Genus One Cases}

This elk model in Fig.(\ref{fig:elk}) is a genus one surface. The left frame shows the quad-mesh $Q$, the right frame shows the holomorphic 1-form $\omega_0$. The induced meromorphic quartic differential $\omega_Q$ has $26$ poles and $26$ zeros,
\[
    D_Q=(\omega_Q)= \sum_{i=1}^{26} (p_i-q_i).
\]
The results of the Abel-Jacobi map are as follows:
\[
\mu\left(\sum_{i=1}^26 p_i \right) = 5.8952427+i2.7571920, \quad
\mu\left(\sum_{j=1}^26 q_j \right) = 5.8954460+i2.7571919.
\]
Therefore, $\mu(D_Q)$ is the difference between them, which is $-2.03332e-04~+~i1.044e-07$, very close to the theoretic prediction $(0,0)$.

\begin{figure*}[h]
\centering
\begin{tabular}{cc}
	\includegraphics[height=0.33\textwidth]{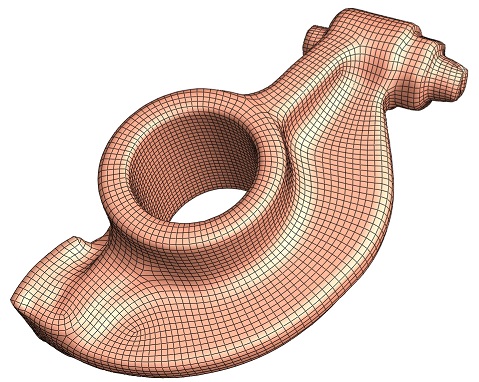} &
	\includegraphics[height=0.33\textwidth]{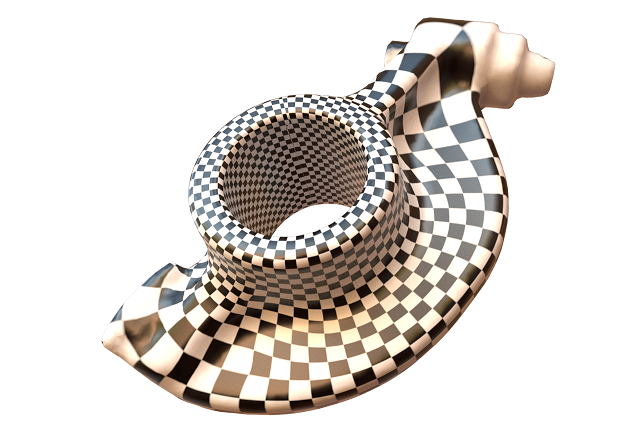} \\
	(a) original quad mesh & (b) holo 1-form  \\
\end{tabular}
\caption{The genus 1 rockerarm model.}
	\label{fig:rockerarm}
\end{figure*}

Fig.~\ref{fig:rockerarm} shows another genus one model, the rocker-arm. The left frame shows the quad-mesh $Q$, the right frame shows the holomorphic 1-form $\omega_0$. The induced meromorphic quartic differential has $18$ poles and $18$ zeros.
\[
    D_Q = \sum_{i=1}^{18}(p_i-q_i),
\]
The results of the Abel-Jacobi map are as follows:
\[
\mu\left(\sum_{j=1}^{18} p_j \right) = 2.61069+i0.588368,\quad
\mu\left(\sum_{i=1}^{18} q_i \right) = 2.61062+i0.588699.
\]
Hence, $\mu(D_Q)$ is the difference between them, which equals to $6.967e-05~-~i3.3064e-4$, very close to the origin in $J(S_Q)$.

\begin{figure}[h!]
\centering
\begin{tabular}{rl}
\includegraphics[height=0.45\textwidth]{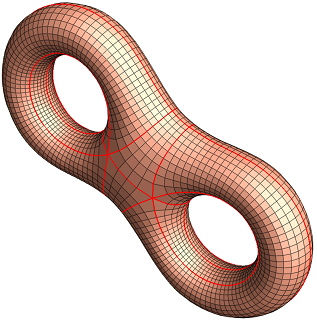}&
\includegraphics[height=0.45\textwidth]{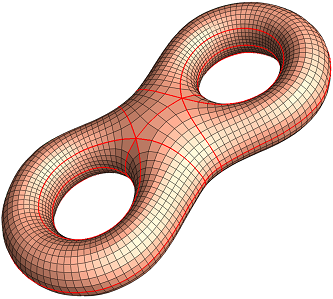}\\
\end{tabular}
\caption{The input genus two quad-mesh.}
\label{fig:genus_two_quad_mesh}
\end{figure}

\begin{figure}[h!]
\centering
\begin{tabular}{rl}
\includegraphics[height=0.485\textwidth]{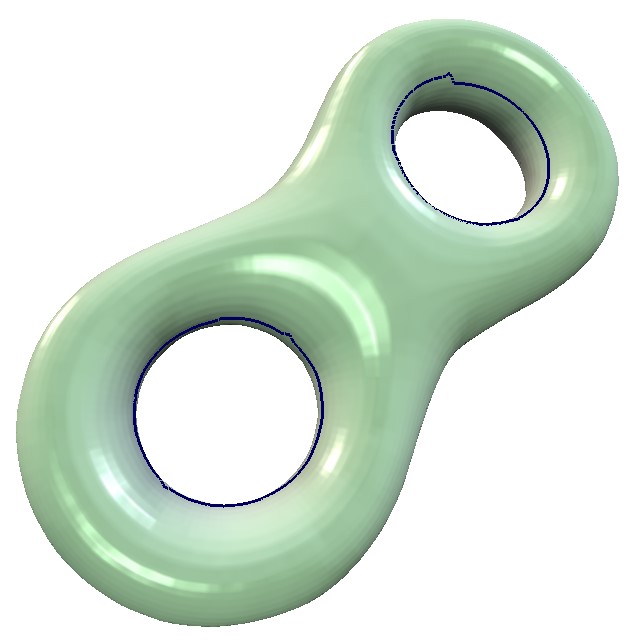}&
\includegraphics[height=0.485\textwidth]{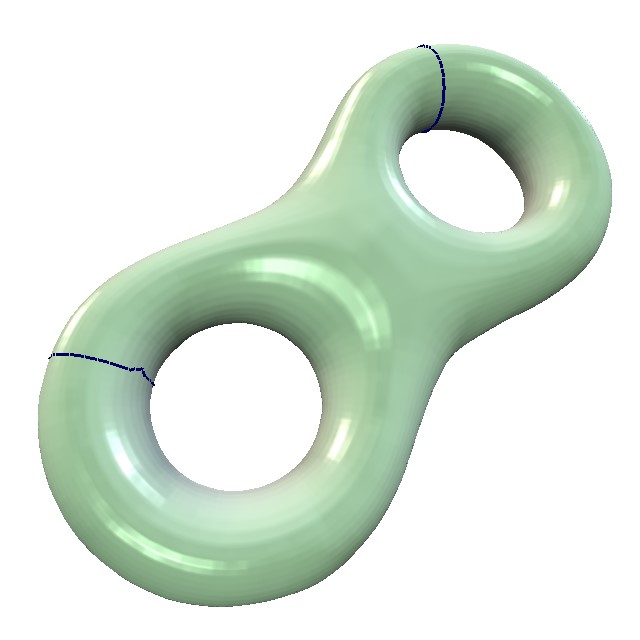}\\
(a) tunnel loops & (b) handle loops\\
\end{tabular}
\caption{The homology group basis.}
\label{fig:genus_two_homology_basis}
\end{figure}

\begin{figure}[h!]
\centering
\begin{tabular}{ll}
\includegraphics[height=0.485\textwidth]{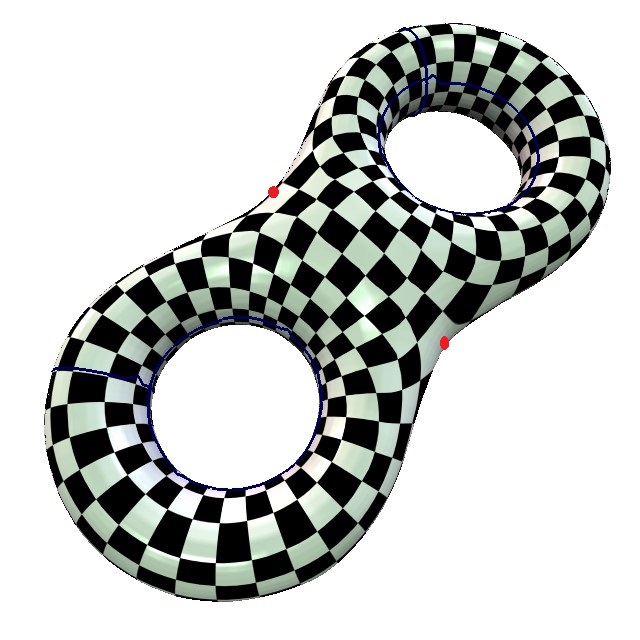}&
\includegraphics[height=0.485\textwidth]{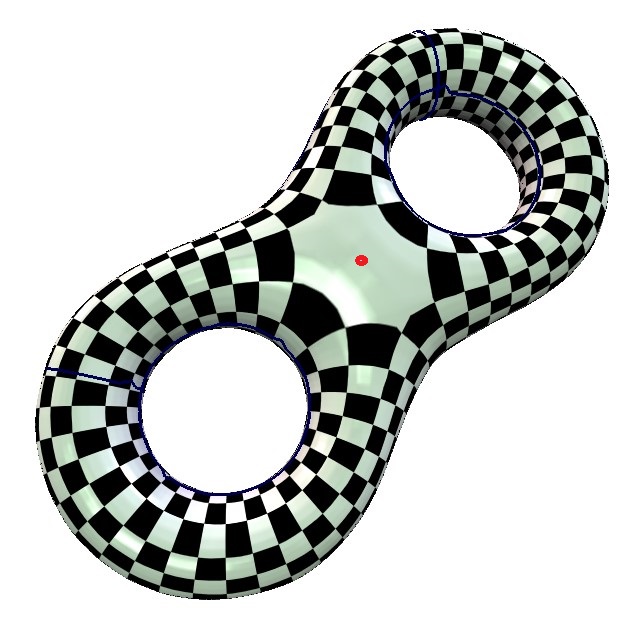}\\
\end{tabular}
\caption{The holomorphic differential basis.}
\label{fig:genus_two_holo_basis}
\end{figure}

\paragraph{Genus Two Case}

Fig.~\ref{fig:genus_two_quad_mesh} shows a genus two quad-mesh $Q$,which has four order two zeros. Fig.~\ref{fig:genus_two_homology_basis} shows the homology group basis of the mesh, the tunnel loops $a_0,a_1$ and the handle loops $b_0,b_1$. Fig.~\ref{fig:genus_two_holo_basis} shows holomorphic differential basis $\varphi_0$ and $\varphi_1$ computed under the quad-mesh metric $\mathbf{g}_Q$. We set $\varphi_0$ as $\omega_0$ and verify the Abel-Jacobi condition by computing the Abel-Jacobi map $\mu(D_Q - 4(\omega_0))$. The period matrix $A$ of the Riemann surface $S_Q$ is
$$
{\footnotesize
\left(
\begin{array}{cc}
0.99999999~-~i1.4209e-09& -0.99999989+i 6.01812e-08 \\
0.99999998+i5.12829e-09 & 0.99999992-i2.88896e-08
\end{array}
\right)
}
$$
The period matrix $B$ is
$$
{\footnotesize
\left(
\begin{array}{cc}
3.18e-08 +i 0.38191542 & 4.7433845e-20 +i 0.3861979 \\
1.433e-08+i 0.44392235 & -2.3716923e-20 -i 0.44820492
\end{array}
\right)
}
$$
The Abel-Jacobi image of the divisor,
$$
\mu(D_Q-4(\omega_0)) =
{\footnotesize
\left(
\begin{array}{c}
1e-06 \\
2e-07 -i 1.6e-06
\end{array}
\right)},
$$

which is very close to $0$. This shows the Abel-Jacobi condition holds for the quad-mesh in Fig.~\ref{fig:genus_two_quad_mesh}.

\begin{figure}[h]
\centering
\begin{tabular}{cc}
	\includegraphics[height=0.45\textwidth]{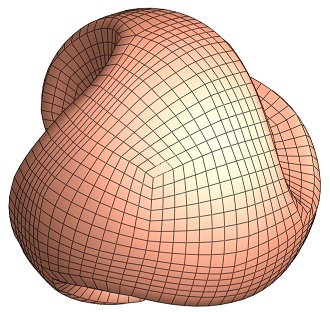} &
\includegraphics[height=0.45\textwidth]{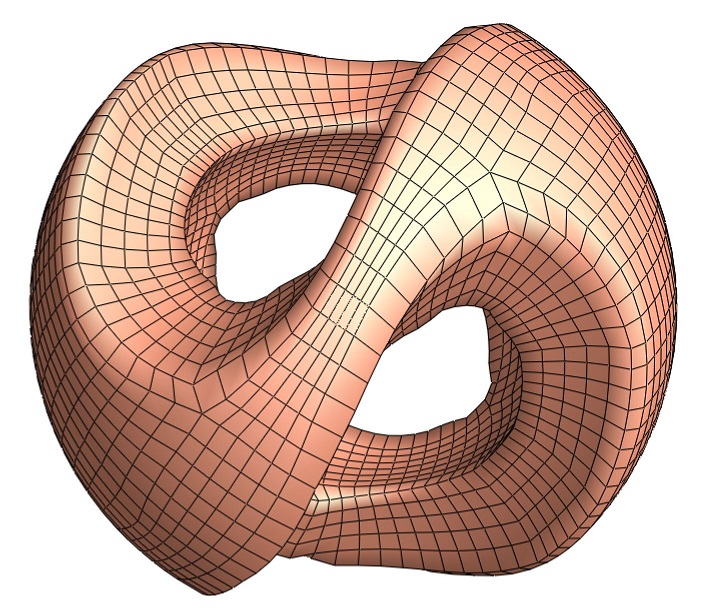}\\
\end{tabular}
\caption{A genus $2$ quad-mesh of a sculpture model.}
	\label{fig:scultpure_quad}
\end{figure}

\begin{figure}[h]
\centering
\begin{tabular}{cc}
	\includegraphics[height=0.45\textwidth]{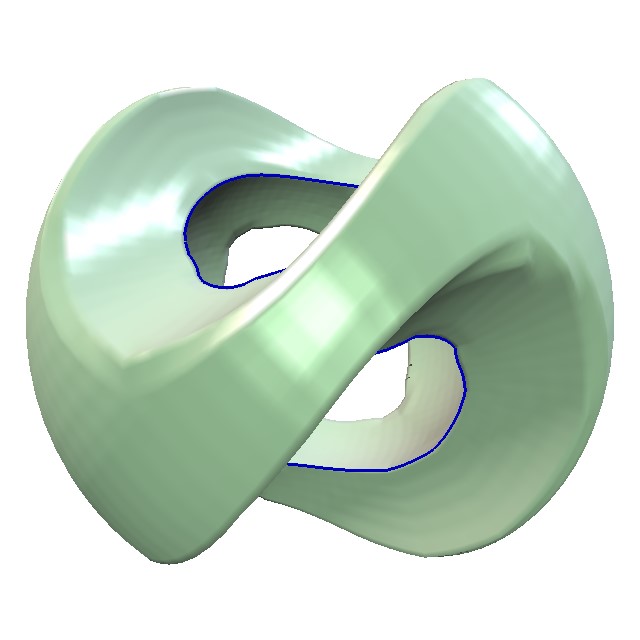} &
\includegraphics[height=0.45\textwidth]{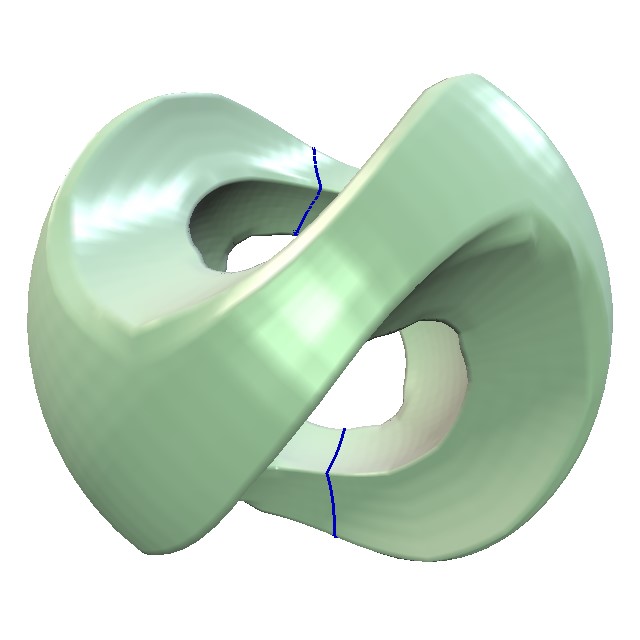}\\
(a) tunnel loops &(b) handle loops\\
\end{tabular}
\caption{Tunnel and handle loops of the sculpture model.}
	\label{fig:scultpure_homology}
\end{figure}

\begin{figure}[h]
\centering
\begin{tabular}{cc}
	\includegraphics[height=0.45\textwidth]{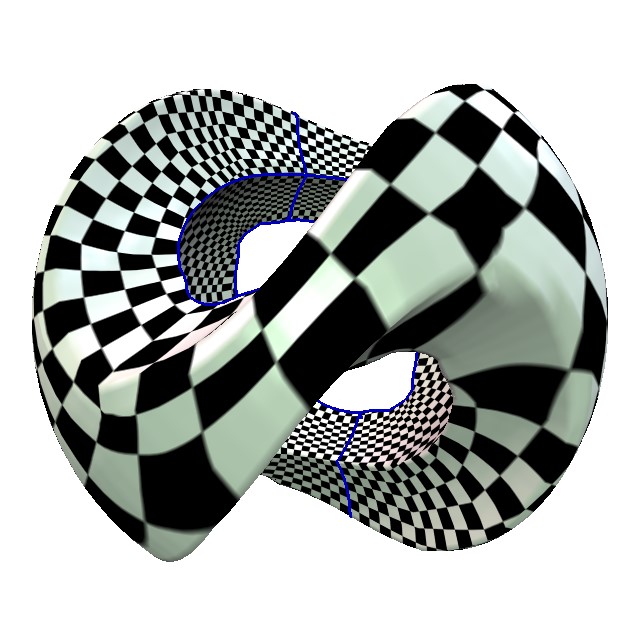} &
\includegraphics[height=0.45\textwidth]{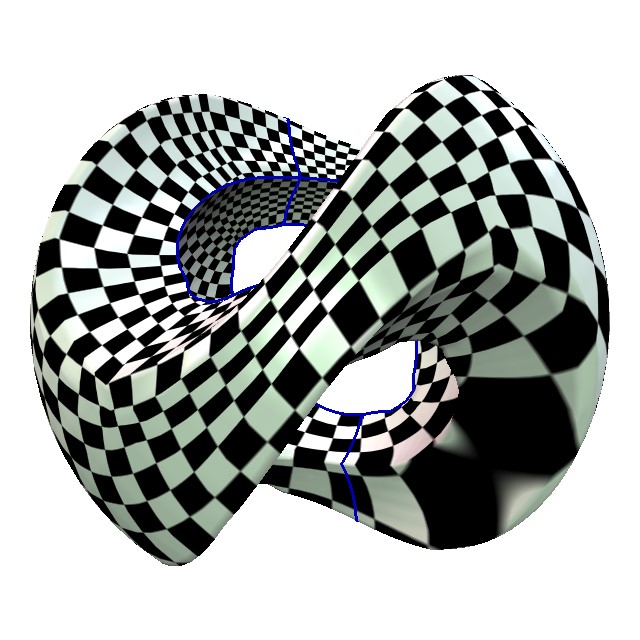}\\
(a) tunnel loops &(b) handle loops\\
\end{tabular}
\caption{Holomorphic 1-form basis of the sculpture model.}
	\label{fig:scultpure_holomorphic}
\end{figure}

Fig.~\ref{fig:scultpure_quad} shows another genus two quad-mesh of a sculpture model which has $12$ valence-5 vertices and $4$ valence $3$ vertices. Fig.~\ref{fig:scultpure_homology} shows the tunnel and handle loops of the sculpture model. Fig.~\ref{fig:scultpure_holomorphic} shows the basis of the holomorphic differentials $\varphi_0$ and $\varphi_1$.

We set $\varphi_0$ as $\omega_0$ and verify the Abel-Jacobi condition by computing the Abel-Jacobi map $\mu(D_Q - 4(\omega_0))$. The period matrix $A$ of the Riemann surface $S_Q$ is
$$
{\footnotesize
\left(
\begin{array}{cc}
0.99999997~-~i2.8e-09& -0.24999994+i 2.745e-08 \\
0.99999999+i1.13e-08 & 0.50000015+i4.1e-08
\end{array}
\right).
}
$$
The period matrix $B$ is
$$
{\footnotesize
\left(
\begin{array}{cc}
-4.8789098e-19 +i 0.50669566 & 7.5894152e-19 +i 0.15720634 \\
-7.5894152e-19+i 0.73261918 & 4.8789098e-19 +i 0.589281
\end{array}
\right).
}
$$
The Abel-Jacobi map image of the divisor is
$$
\mu(D_Q-4(\omega_0)) =
{\footnotesize
\left(
\begin{array}{c}
-1.568599999979e-05+i3.69999999994e-06 \\
4.28899999998e-05 -i 4.400000000182e-07
\end{array}
\right),
}
$$
which is very close the $0$.

%

\subsection{Meromorphic Quartic Differential Construction}
\label{sec:algorithm}

In practice, a surface $\Sigma$ embedded in $\mathbb{R}^3$ is given, the induced Euclidean metric is $\mathbf{g}$. Our purpose is to construct a quadrilateral mesh $\mathcal{Q}$ on $(\Sigma,\mathbf{g})$. It is highly desirable that the metric induced by $\mathcal{Q}$, $\mathbf{g}_Q$, is conformal equivalent to the original metric $\mathbf{g}$. From above discussion, we see the equivalence between a quad-mesh $\mathcal{Q}$ and the meromorphic quartic differential $\omega_Q$ on the Riemann surface $S_Q$. Therefore $\omega_Q$ is also a meromorphic differential on $(\Sigma,\mathbf{g})$.

%
%
%

Suppose $(\Sigma,\mathbf{g})$ is a genus zero closed surface, then it is conformal equivalent to the unit sphere $\mathbb{S}^2$, namely $\mathbb{C}\cup \{\infty\}$. The unit sphere $\mathbb{S}^2$ has two charts $z$ and $w$, $z$ covers $\mathbb{C}$, $w$ covers $\mathbb{C}\cup\{\infty\}\setminus \{0\}$, the transition map is $w=1/z$. We define the simplest meromorphic differential
\[
    \omega_0 = dz = -\frac{1}{w^2} dw.
\]
Given a quad-mesh $\mathcal{Q}$ on $\Sigma$, then $f=\omega_Q / w_0^4$ is a meromorphic function. On the sphere, any meromorphic function is a rational function, the general global representation of $\omega_Q$ on $\mathbb{C}$ is given by:
\begin{equation}
    \omega_Q = \frac{(z-p_1)(z-p_2)\cdots (z-p_{n-8})}{(z-q_1)(z-q_2)\cdots(z-q_n)}dz^4,
    \label{eqn:meromorphic_differential_0}
\end{equation}
where $\{p_1,p_2,\dots,p_{n-8}\}$ are the simple zeros, $\{q_1,q_2,\dots,q_n\}$ are the simple poles of $\omega_Q$. The $\infty$ point is the zero of order $8$. Each simple zero corresponds to a valence $5$ vertex, each simple pole corresponds to a valence $3$ vertex. The zeros (poles) can be merged into high order ones, thus Eqn.~\ref{eqn:meromorphic_differential_0} becomes
\begin{equation}
    \omega_Q = \frac{\Pi_{i=1}^k (z-p_i)^{n_i}}{\Pi_{j=1}^l (z-q_j)^{m_j}}dz^4,
    \label{eqn:meromorphic_differential_0_1}
\end{equation}
where $\sum_{j=1}^l m_j - \sum_{i=1}^k = 8$, $m_i$'s and $n_j$'s are positive integers.

Eqn.~\ref{eqn:meromorphic_differential_0} and Eqn.~\ref{eqn:meromorphic_differential_0_1} give all possible meromorphic quartic differentials on the sphere, some of them have infinite trajectories. The ones with finite trajectories correspond to $\omega_Q$ for some quad-mesh $Q$.

\begin{figure*}[h]
\centering
\begin{tabular}{cc}
	\includegraphics[height=0.45\textwidth]{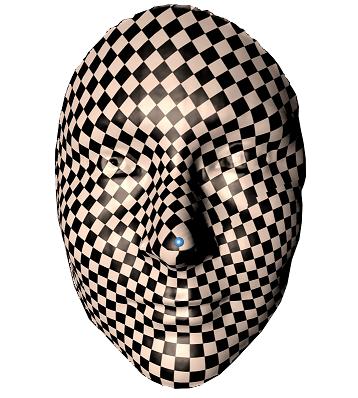} &
	\includegraphics[height=0.45\textwidth]{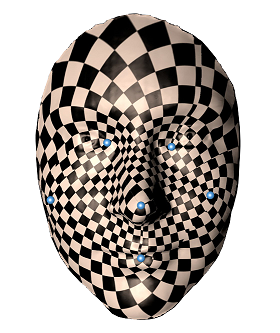}\\
	(a) 1 pole & (b) 6 poles \\
\end{tabular}
\caption{The visualization of meromorphic quartic differential on a face model.}
	\label{fig:sophie}
\end{figure*}

Fig.~(\ref{fig:sophie}) left frame shows two different meromorphic quartic differentials constructed in this way. In the left frame, we compute a Riemann mapping to conformally map the facial surface $S$ onto the planar unit disk using the method in \cite{Gu_JDG_2018}, then pick a pole $z_1==0.706853+0.52086i$. The meromorphic quartic differential is given by
\[
    \omega = \frac{1}{z-z_1} (dz)^4.
\]
The we find a path $\gamma$ from $z_1$ to the boundary, slice the surface along $\gamma$ to get a simply connected domain $\bar{S}$. In this domain, we choose one branch of $\sqrt[4]{\omega}$. By integrating $\sqrt[4]{\omega}$ on $\bar{S}$, we map $\bar{S}$ onto the complex plane. Then we use checkerboard texture mapping to visualize the trajectories of $\omega$.

Fig.~(\ref{fig:sophie}) right frame demonstrates a meromorphic quartic differential with $6$ poles,
\begin{equation}
    \omega \frac{1}{\prod_{i=1}^{6}(z-z_i)}(dz)^4
    \label{eqn:sophie_differential}
\end{equation}
where the poles are given in Tabl.~\ref{tab:poles_sophie}.

\begin{table}[h!]
{\footnotesize
\begin{tabular}{lll}
$z_1=0.451559+0.21962i$&$z_2=0.45696+0.617636i$&$z_3=0.706853+0.52086i$\\
$z_4=0.533522+0.407822i$&$z_5=0.250598+0.471244i$&$z_6=0.747474+0.28336i$
\end{tabular}
\caption{Poles in the meromorphic differential Eqn.~\ref{eqn:sophie_differential}.\label{tab:poles_sophie}}
}
\end{table}
We compute a cut graph $\gamma$ connecting all the singularities and the boundary using the algorithm in \cite{adams1995hitchhiker}, then slice the surface along $\gamma$ to get a simply connected domain $\bar{S}$. By integrating a branch of $\sqrt[4]{\omega}$, we map $\bar{S}$ onto the complex plane.

\begin{figure*}[h]
\centering
\begin{tabular}{cc}
\includegraphics[height=0.45\textwidth]{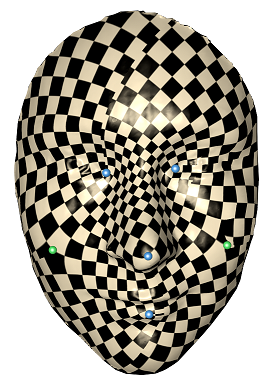}&
\includegraphics[height=0.45\textwidth]{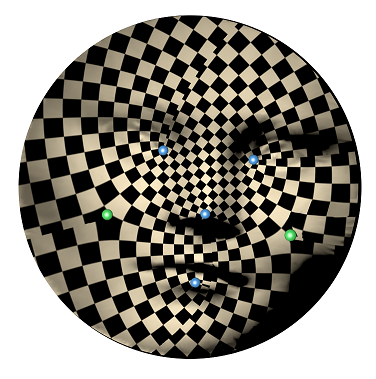}\\
	(a) 4 holes and 2 zeros & (b) Riemann mapping of (a)\\
\end{tabular}
\caption{The visualization of meromorphic quartic differential on a face model.}
	\label{fig:sophie_2}
\end{figure*}

Fig.~(\ref{fig:sophie_2}) shows a meromorphic quartic differential with $4$ poles and $2$ zeros,
\begin{equation}
    \omega \frac{(z-p_1)(z-p_2)}{\prod_{i=1}^{6}(z-q_i)}(dz)^4,
    \label{eqn:sophie_differential_2}
\end{equation}
where the poles and zeros are given in Tab.~\ref{tab:sophie_differential_2}.
\begin{table}[h!]
{\footnotesize
\begin{tabular}{lll}
$p_1=0.250598+0.471244i$ & $p_2=0.747474+0.28336i$ & $q_1=0.451559+0.21962i$  \\ $q_2=0.45696+0.617636i$  & $q_3=0.706853+0.52086i$  & $q_4=0.533522+0.407822i$ \\
\end{tabular}
\caption{The zeros and poles in the meromorphic differential in Eqn.~\ref{eqn:sophie_differential_2}.\label{tab:sophie_differential_2}}
}
\end{table}

\begin{figure*}[h]
\centering
\begin{tabular}{ccc}
	\includegraphics[height=0.33\textwidth]{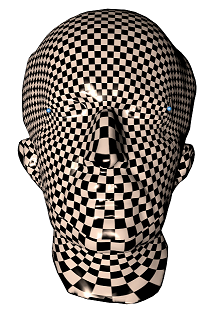} &
	\includegraphics[height=0.33\textwidth]{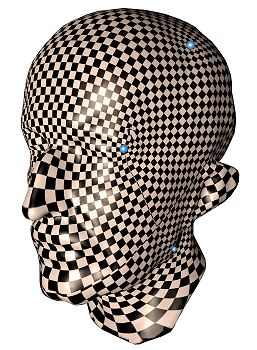} &
	\includegraphics[height=0.33\textwidth]{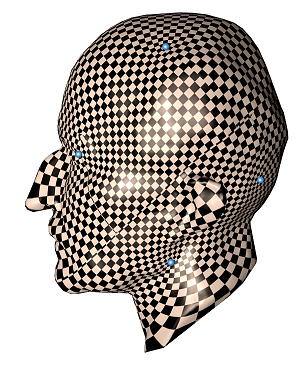} \\
	\includegraphics[height=0.33\textwidth]{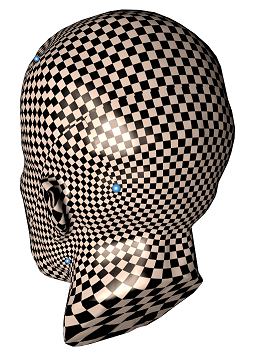} &
	\includegraphics[height=0.33\textwidth]{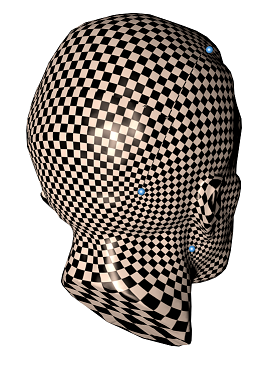} &
	\includegraphics[height=0.33\textwidth]{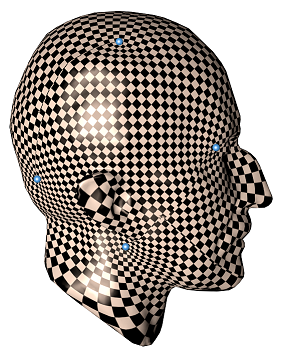}\\
\end{tabular}
\caption{A meromorphic quartic differential on the Max-Planck sculpture model with $8$ poles.}
	\label{fig:max_planck}
\end{figure*}

\begin{figure*}[h]
\centering
\begin{tabular}{cccc}
	\includegraphics[height=0.23\textwidth]{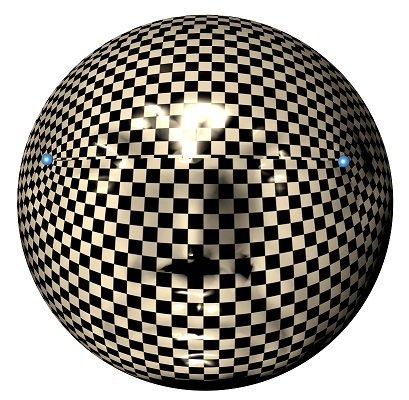} &
\includegraphics[height=0.23\textwidth]{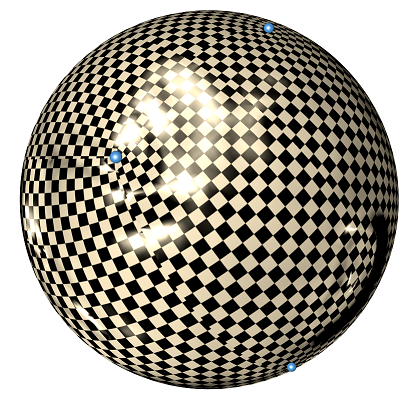} &
\includegraphics[height=0.23\textwidth]{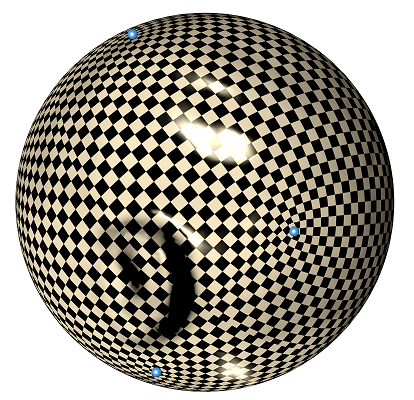} &
\includegraphics[height=0.23\textwidth]{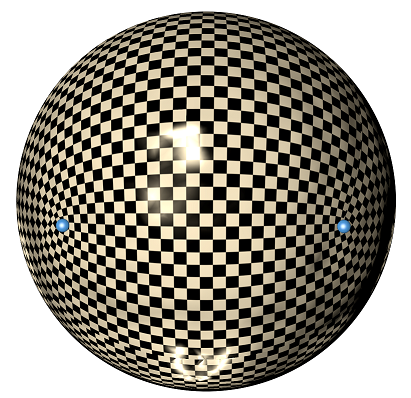} \\
\includegraphics[height=0.23\textwidth]{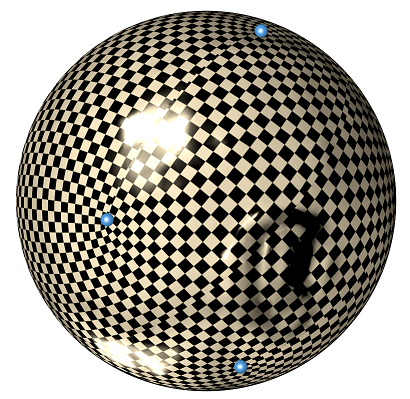} &
\includegraphics[height=0.23\textwidth]{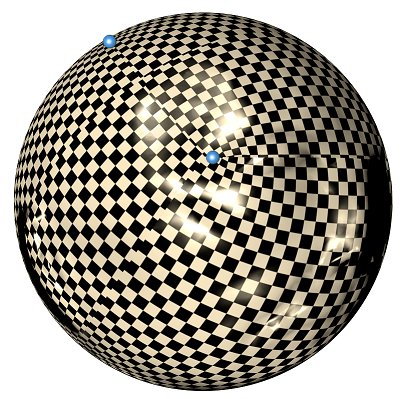} &
\includegraphics[height=0.23\textwidth]{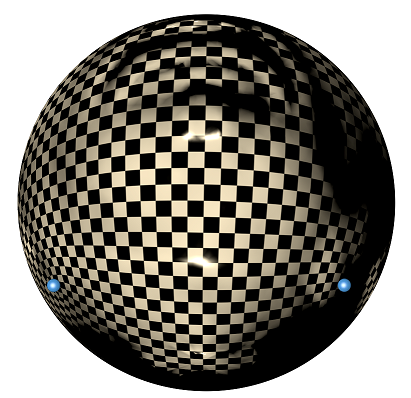} &
\includegraphics[height=0.23\textwidth]{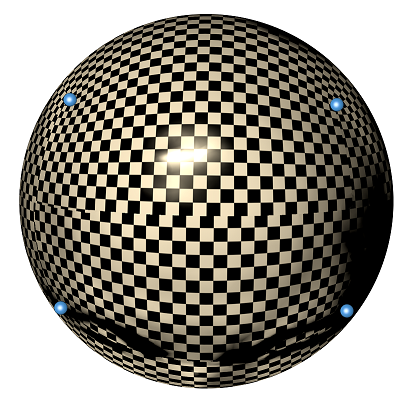}
\end{tabular}
\caption{The meromorphic quartic differential the conformal spherical image of the Max-Planck sculpture model with $8$ poles.}
	\label{fig:max_planck_spheres}
\end{figure*}

Fig.~\ref{fig:max_planck} demonstrates a meromorphic quartic differential on the Max Planck head model. First, we conformally map the model onto the unit sphere using the algorithm in \cite{DBLP:journals/tmi/GuWCTY04}, then the sphere is mapped onto the complex plane using the stereo-graphic projection. On the complex plane, we construct a meromorphic quadratic differential as
\begin{equation}
    \omega = \frac{1}{\prod_{i=1}^{8}(z-q_i)}(dz)^4,
    \label{eqn:max_planck_differential}
\end{equation}
where the poles are given in Tab.~\ref{tab:max_planck_differential}.

\begin{table}[h!]
{\footnotesize
\begin{tabular}{llll}
$q_1=1.32607+1.3106i$&$q_2=-1.27859+1.27903i$&$q_3=1.30017-1.25335i$\\
$q_4=-1.29695-1.28728i$&$q_5=0.471821-0.46131i$&$q_6=-0.443743-0.468551i$\\
$q_7=0.452511+0.463833i$&$q_8=-0.450766+0.468879i$&
\end{tabular}
\caption{The poles of the meromorphic differential on Max Planck head model in Eqn.~\ref{eqn:max_planck_differential}.\label{tab:max_planck_differential}}
}
\end{table}

Similarly, a cut graph $\gamma$ connecting all singularities is computed, the surface is sliced along the cut graph to get a simply connected domain. We integrate $\sqrt[4]{\omega}$ on the domain, to flatten the surface to the complex plane. By checkerboard texture mapping, the trajectories of $\omega$ can be visualized.

\begin{figure*}[h]
\centering
\begin{tabular}{ccc} \includegraphics[height=0.33\textwidth]{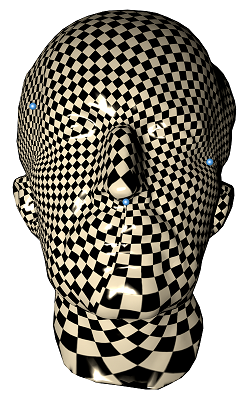}& \includegraphics[height=0.33\textwidth]{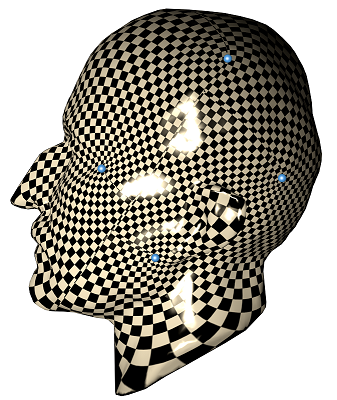}& \includegraphics[height=0.33\textwidth]{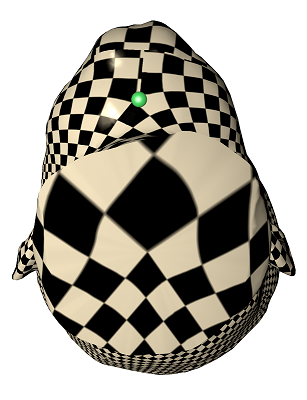}\\
\includegraphics[height=0.33\textwidth]{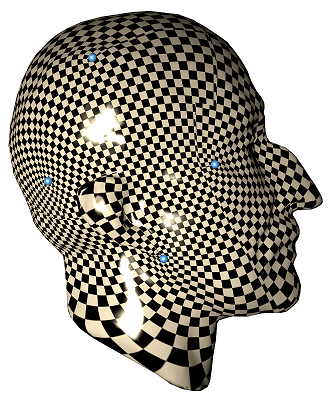}&	\includegraphics[height=0.33\textwidth]{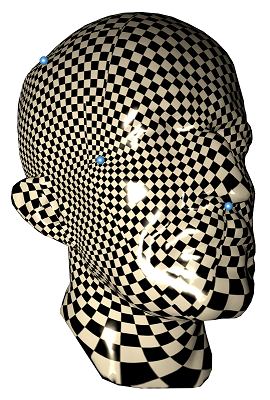}&	\includegraphics[height=0.32\textwidth]{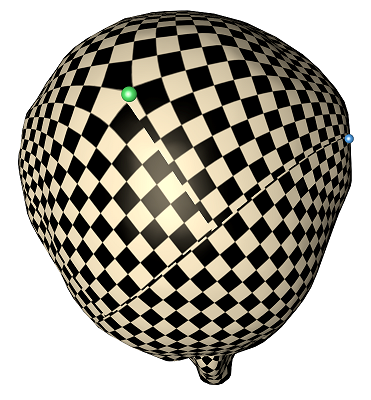}\\
\end{tabular}
\caption{A meromorphic quartic differential on the Max-Planck sculpture model with $10$ poles (blue) and $2$ zeros (green).\label{fig:max_planck_2}}
\end{figure*}

\begin{figure*}[h]
\centering
\begin{tabular}{cccc} \includegraphics[height=0.23\textwidth]{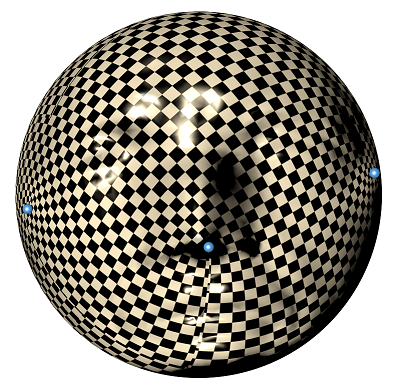} &
\includegraphics[height=0.23\textwidth]{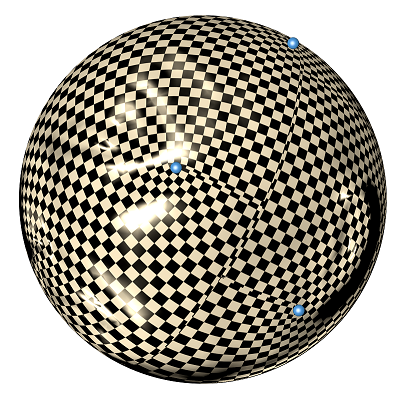} &
\includegraphics[height=0.23\textwidth]{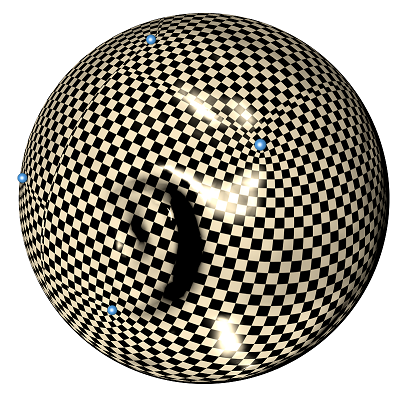} &
\includegraphics[height=0.23\textwidth]{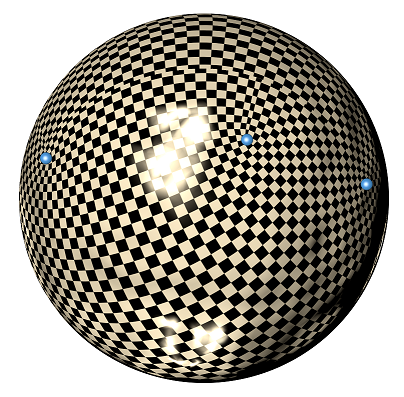} \\
\includegraphics[height=0.23\textwidth]{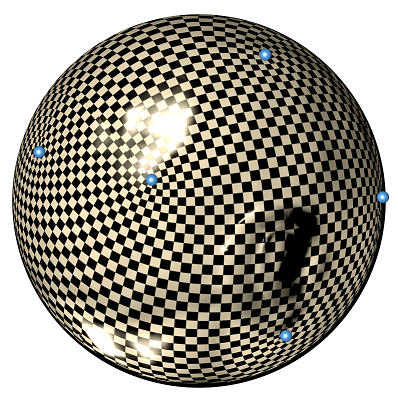} &
\includegraphics[height=0.23\textwidth]{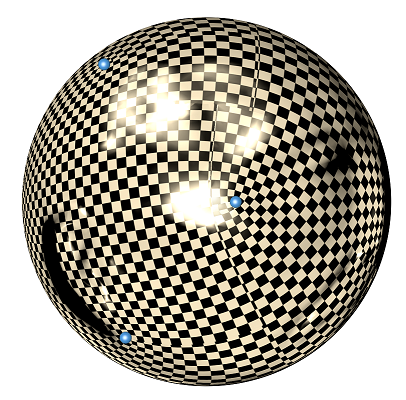} &
\includegraphics[height=0.23\textwidth]{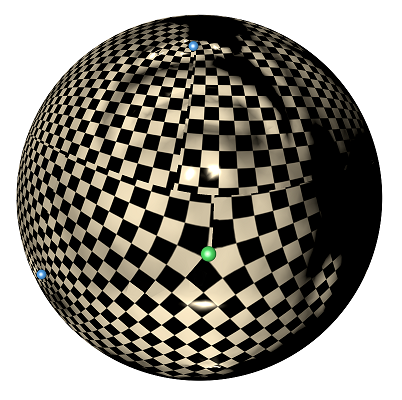} &
\includegraphics[height=0.23\textwidth]{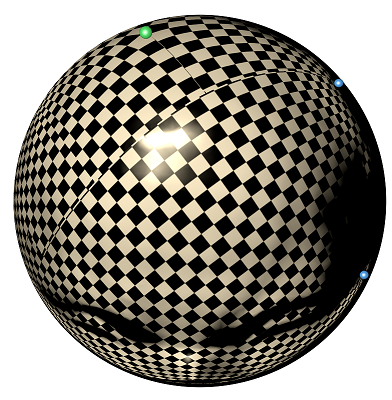}
\end{tabular}
\caption{The meromorphic quartic differential the conformal spherical image of the Max-Planck sculpture model with $10$ poles and $2$ zeros.}
	\label{fig:max_planck_spheres_2}
\end{figure*}

The second meromorphic quartic differential is shown in Fig.~\ref{fig:max_planck_2}, which has the form on the complex plane
\begin{equation}
    \omega = \frac{\prod_{i=1}^2 (z-p_i)}{\prod_{j=1}^8(z-q_j)} (dz)^4,
    \label{eqn:max_planck_differential_2}
\end{equation}
where the zeros and the poles are given in Tab.~\ref{tab:max_planck_2}.

\begin{table}[h!]
{\footnotesize
\begin{tabular}{lll}
$p_1=0.898261+3.24367i$& $p_2=-0.00810208-0.25253i$ & $q_1=0.00289177+0.255035i$ \\ $q_2=1.03926-3.32305i$ & $q_3=1.5921+0.915034i$&$q_4=-1.61079+0.858346i$\\
$q_5=1.61098-0.845895i$& $q_6=-1.63865-0.894138i$ & $q_7=0.559829+0.296053i$\\
$q_8=-0.564592+0.307631i$ & $q_9=0.555884-0.307683i$&$p_{10}=-0.550573-0.311611i$
\end{tabular}
\caption{The zeros and poles of the meromorphic differential on the Max Planck head model in Eqn.~\ref{eqn:max_planck_differential_2}.\label{tab:max_planck_2}}}
\end{table}

\section{Conclusion}
\label{sec:conclusion}

This work proves the equivalence between quadrilateral meshes and meromorphic quartic differentials on Riemann surfaces with finite trajectories (theorem \ref{thm:differential_quad} and \ref{thm:differential_quad}); Second, this work gives Abel-Jacobi condition for the configurations of singularities of quad-meshes (theorem \ref{thm:Abel_Jacobian_condition}), the condition can be easily verified algorithmically; Third, the meromorphic quartic differentials can be constructed on surfaces using their global algebraic representation, this leads to a novel direction for quad-mesh generation based on meromorphic differentials. Our experimental results demonstrate that the method is theoretically rigorous, practically simple and efficient. This opens a new direction for quad-mesh generation based on Riemann surface theory.

In future, we will explore the methods to guarantee the finiteness of all the trajectories of meromorphic differentials, the algorithm for divisor optimization to satisfy the Abel-Jacobi condition and generalize the Abel-Jacobi condition to hexahedral meshes.

\bibliographystyle{plain}
\bibliography{references,quad}

\end{document}